\documentclass[a4paper]{article}
\usepackage[utf8]{inputenc}
\usepackage[english]{babel} % English language/hyphenation
\usepackage{amsmath,amssymb,amsthm,bm} % Math packages
\usepackage[margin=2cm]{geometry}
\usepackage[color=yellow]{todonotes}
\usepackage{esint}
\usepackage[colorlinks]{hyperref}
\usepackage{enumerate}
\usepackage{outlines}[enumerate]
\usepackage{framed,comment}
\usepackage{upgreek}
\usepackage{pgfplots}
\usepackage{float}
\usepackage{tikz,tikz-cd}
\usepackage{rotating}
\usepackage{graphicx}
\usepackage{caption}
\usepackage{subcaption}
\usepackage{appendix}
\usepackage{pgf, pgffor}
\extrafloats{100}

\newtheorem{theorem}{Theorem}[section]

\newtheorem{proposition}[theorem]{Proposition}
\newtheorem{remark}[theorem]{Remark}

\usepackage{parskip}
\numberwithin{equation}{section}

\newcommand{\bs}[1]{\boldsymbol{#1}}

\pgfplotsset{compat=1.14}

\title{Nonlinear dispersion in wave-current interactions}
\author{Darryl D. Holm\quad\hbox{and}\quad Ruiao Hu \\ \smallskip\large
	Department of Mathematics\\ 
	Imperial College London SW7 2AZ, UK
	\\ \bigskip\normalsize
	d.holm@ic.ac.uk and ruiao.hu15@ic.ac.uk} 
\date{In honour of Tony Bloch's 65th birthday. Happy birthday, Tony! }                                          % Activate to display a given date or no date

\begin{document}
	
\maketitle

\begin{abstract}
Via a sequence of approximations of the Lagrangian in Hamilton's principle for dispersive nonlinear gravity waves we derive a hierarchy of Hamiltonian models for describing wave-current interaction (WCI) in nonlinear dispersive wave dynamics on free surfaces. A subclass of these WCI Hamiltonians admits \emph{emergent singular solutions} for certain initial conditions. These singular solutions are identified with a singular momentum map for left action of the diffeomorphisms on a semidirect-product Lie algebra. This semidirect-product Lie algebra comprises vector fields representing horizontal current velocity acting on scalar functions representing wave elevation. We use computational simulations to demonstrate the dynamical interactions of the emergent wavefront trains which are admitted by this special subclass of Hamiltonians for a variety of initial conditions.\\

\noindent
In particular, we investigate: 
\\(1) A variety of localised initial current configurations in still water whose subsequent propagation generates surface-elevation dynamics on an initially flat surface; and 
\\(2) The release of initially confined configurations of surface elevation in still water that generate dynamically interacting fronts of localised currents and wave trains. 
\\The results of these simulations show intricate wave-current interaction patterns whose structures are similar to those seen, for example, in Synthetic Aperture Radar (SAR) images taken from the space shuttle. 
\end{abstract}

%\tableofcontents

% \setcounter{tocdepth}{1}
% \tableofcontents

%\input{Introduction}
%\input{SALT}
%\input{SFLT}
%\input{EA_SFLT}
%\input{Examples}

%\section{Conclusion and outlook}\label{sec: conclude}

%\todo[inline]{DH: (1) Discuss physics and highlight the energy-preserving 3D Euler fluid equation which also preserve helicity, 

%(2) compare with M\'emin LU approach. 

%(3) say more about wind-driven stochastic vortex forcing of Langmuir turbulence.}

\section{Introduction}\label{sec-1}

The sea-surface disturbances whose trains of curved wavefronts trace the propagation of internal gravity waves on the ocean thermocline hundreds of meters below the surface may be observed in many areas of strong tidal flow.
For example, the passage of the Atlantic Ocean tides through the Gibraltar Strait produces trains of curved sea-surface wavefronts expanding into the Mediterranean Sea. Likewise, the passage of the Pacific Ocean tides through the Luzon Strait between Taiwan and the Philippines produces trains of curved sea-surface wavefronts expanding into the South China Sea.  These coherent trains of expanding curved wavefront disturbances are easily observable because they are strongly nonlinear. Their sea-surface signatures in the South China Sea may even be seen from the Space Shuttle  \cite{BLLH2017,LCHL1998,ZLL2014}, as prominent crests of wave trains move in great arcs hundreds of kilometres in length and traverse sea basins thousands of kilometres across. Figures \ref{fig: Gibralter} and \ref{fig: South China Sea} show SAR images of the signatures of internal waves on sea surface.
    \begin{figure}[H]
        \centering
        \begin{subfigure}[b]{0.48\textwidth}  
			\centering 
			\includegraphics[width=\textwidth, height = \textwidth]{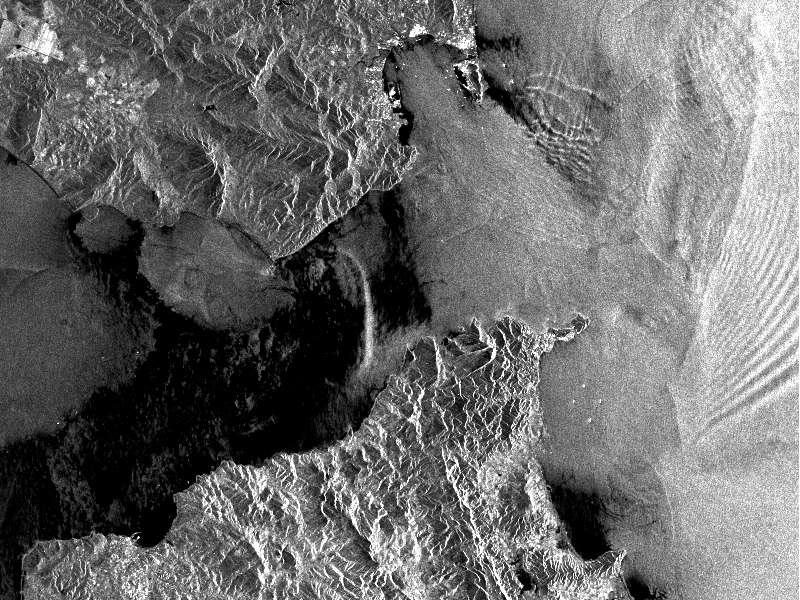}
		\end{subfigure}
	    \begin{subfigure}[b]{0.48\textwidth}  
			\centering 
			\includegraphics[width=\textwidth, height = \textwidth]{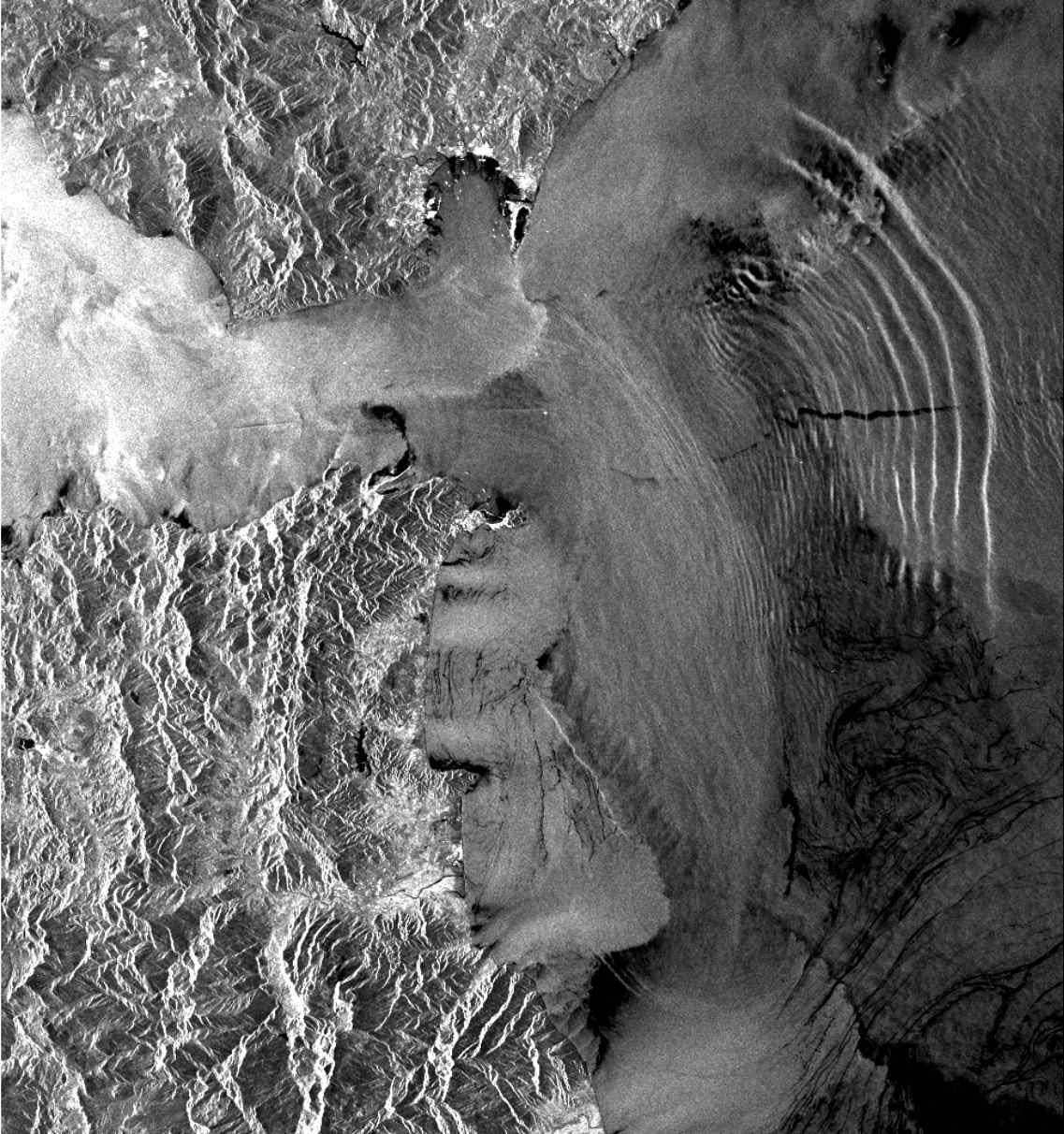}
		\end{subfigure}
        \caption{Synthetic Aperture Radar (SAR) Image of internal-wave signatures near Gibraltar Strait
        Taken from \url{https://earth.esa.int/web/guest/missions/esa-operational-eo-missions/ers/instruments/sar/applications/tropical/-/asset_publisher/tZ7pAG6SCnM8/content/oceanic-internal-waves}}
    \label{fig: Gibralter}
    \end{figure}
    \begin{figure}[H]
        \centering
		\includegraphics[width=\textwidth, height = .75\textwidth]{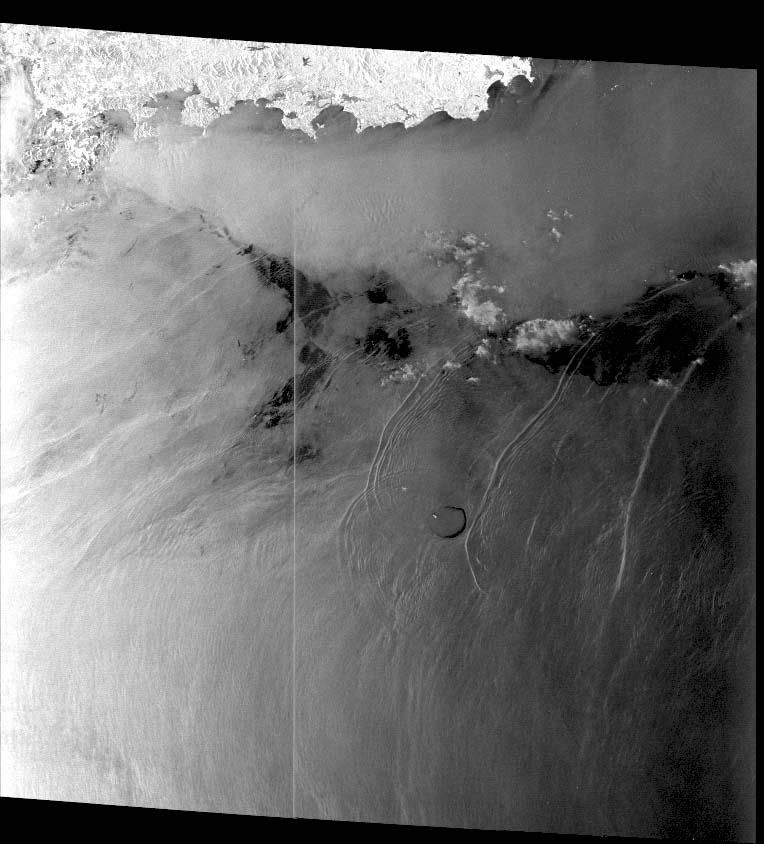}
        \caption{Synthetic Aperture Radar (SAR) Image of internal-wave signatures on the South China Sea. Taken from \url{https://earth.esa.int/web/guest/missions/esa-operational-eo-missions/ers/instruments/sar/applications/tropical/-/asset_publisher/tZ7pAG6SCnM8/content/oceanic-internal-waves}}
    \label{fig: South China Sea}
    \end{figure}

\paragraph{Multi-layer modelling of internal waves.}
The emission of surface effects near the Gibralter Strait observed in the SAR images seen in figure \ref{fig: Gibralter} are short term effects. In contrast, the surface effects seen in SAR images near Dong Sha Atoll in the South China Sea in figure \ref{fig: South China Sea} are long term effects involving internal wave propagation over hundreds of kilometers. The short term behaviour of internal waves has been modelled with some success using the well known multi-layer Green-Naghdi (MGN) equations \cite{JC2002}. However, longer term modelling of these waves has been problematic, because MGN and its rigid-lid version, the Choi-Camassa (CC) equation \cite{CC1996,CC1999}, were both shown in \cite{LW1997} to be \emph{ill-posed} in the presence of either bathymetry or shear. For example, even the shear induced by a single travelling wave causes the linear growth-rate of a perturbation of MGN or CC solutions to increase without bound as a function of wave number. 

Until recently, the ill-posedness of MGN or CC solutions had prevented convergence under grid refinement of the numerical simulations of these waves over long times, because the cascade of energy to smaller scales would eventually build up at the highest resolved wave number. Regularisation was possible by keeping higher-order terms in an asymptotic expansion, as in for example \cite{BC2009}. However, such methods tended to destroy the Hamiltonian property of the system and also degrade its travelling wave properties. Moreover, if one is to consider the problem of wave generation and propagation at sea, one must consider the effects of bathymetry and shear, both of which may induce instability. Thus, the MGN equations had to be modified to make them well-posed. A recent review of the various approaches to regularising the MGN is given in \cite{D2021}. The analysis in \cite{D2021} focuses on the \emph{Camassa-Holm regime} of asymptotic expansion for nonlinear shallow water waves defined in \cite{CL2009}. The present paper also focuses on this asymptotic expansion. 

\noindent
The multi-layer square-root-${D}$ ($ML\sqrt{D}$) system of nonlinear wave equations introduced in \cite{CHP2010} satisfies three fundamental properties that would be desired in a viable regularisation of MGN. Namely, 
\begin{enumerate}
\item $ML\sqrt{D}$ is linearly well-posed and Hamiltonian; 
\item $ML\sqrt{D}$ preserves the MGN linear dispersion relation for fluid at rest; 
\item $ML\sqrt{D}$ travelling wave solutions agrees with the MGN and KdV $sech^2$ travelling wave form in the absence of imposed background shear. 
\end{enumerate}
Thus, the $ML\sqrt{D}$ Hamiltonian system remains well-posed in the presence of shear and its solutions agree with those of the MGN system in the absence of shear \cite{CHP2010}. With these properties in mind, we shall choose the $ML\sqrt{D}$ Hamiltonian system as the basis for the present work. 

\paragraph{Aims of the present paper.} The overall aim of the present paper is to model the internal-wave surface signatures seen by SAR images such as those in figures \ref{fig: Gibralter} and \ref{fig: South China Sea}. For this purpose, the investigations of the present paper will focus on the theoretical and computational simulation properties of the solutions of the \emph{single-layer} case of $ML\sqrt{D}$, known as $1L\sqrt{D}$. The $1L\sqrt{D}$ model possesses three well-known variants. These are the two-component Camassa-Holm equation (CH2) and the modified CH2 equation (ModCH2) with $H^1$ and $H_{div}$ kinetic energy norms. We will derive these variants and then focus computational simulations on the ModCH2 equation with the $H^1$ kinetic energy norm, which relates to previous work in \cite{FHT2001,FHT2002,HOT2009,HS2013,HT2009}. 

We are inspired by the Synthetic Aperture Radar (SAR) images of the internal-wave signatures of wavefronts on the sea surface shown in figure \ref{fig: Gibralter} and figure \ref{fig: South China Sea}. As mentioned earlier, these wavefronts are known to be driven by internal waves propagating on the interfaces of the stratified layers lying beneath the sea surface \cite{BLLH2017,LCHL1998,ZLL2014}. However, the SAR data only contains the wavefront signatures of the internal waves on the sea surface, as seen from a distance overhead by the Space Shuttle, for example. This means the below-surface processes of their formation cannot be directly observed. To describe the interactions among these wavefront signatures on the surface, we seek a \emph{minimal description} of their dynamics which involves only observable quantities. This minimal model is based on the single-layer version of $ML\sqrt{D}$ which accounts for both kinetic and potential energy. Specifically, we seek to model the formation and dynamics of trains of wavefronts arising from an initial impulse of momentum, or from an initial gradient of surface elevation. We also seek to derive the dynamics of their collisions, including their nonlinear reconnections. In fact, the model we seek would treat the data only as the motion of curves in two dimensions which make optimal use of their kinetic and potential energy over a certain horizontal interaction range. In particular, the minimal model would not attempt to describe the interactions among internal-waves beneath the surface which are believed to produce these wavefronts. 

To formulate such a minimal model of wavefront dynamics, we will derive a sequence of approximate equations in the so-called \emph{Camassa-Holm regime} of nonlinear wave dynamics \cite{D2021}. Starting from the single-layer case ($1L\sqrt{D}$) we will derive the 2D version of the two-component Camassa-Holm equation (CH2). The 1D version of CH2 is well-known for its completely integrable Hamiltonian properties. However, here we will be working in 2D. From CH2, we will obtain the modified two-component Camassa-Holm equation (ModCH2). In 1D, ModCH2 possesses emergent weak solutions supported on points moving along the real line \cite{HT2009,HOT2009}. In the 2D doubly periodic planar case treated here, ModCH2 possesses emergent weak solutions supported on smooth curves embedded in 2D. However, the simulations here do not always capture the singular solutions, indicating that the formation of the singular solution may occur quite slowly. The moving curves in the simulations are meant to model the dynamics of the sea-surface signature wavefronts driven by internal waves interacting below, as seen in the SAR image data. 

Computational simulations of ModCH2 in 2D will be used here to display the various types of interactions among these emergent wave profiles in 2D. These simulations show wave trains with both singular and non-singular profiles emerging from smooth initial conditions. This emergence is followed by reconnections among the wavefronts during their nonlinear interactions. Some of the intricate patterns seen during these 2D simulations turn out to be strikingly similar to those seen in the SAR images shown in figures \ref{fig: Gibralter} and \ref{fig: South China Sea}. 

\paragraph{Plan of the paper.} The plan of the paper is as follows.

Section \ref{sec-1} has provided the problem statement and the main goal of the paper. Namely, we aim to formulate a minimal model of the dynamical wavefront behaviour seen in SAR images such as those shown in figures \ref{fig: Gibralter} and \ref{fig: South China Sea}. We have listed the desired aspects of such a model. These desiderata have already been accomplished in deriving the $ML\sqrt{D}$ model, which provides a multi-layer well-posed description of internal waves in \cite{CHP2010}. Thus, this section has set the context for what follows in the remainder of the paper's investigation of single-layer wavefront interaction dynamics.

Section \ref{sec-2} begins by showing that a certain approximation of the $1L\sqrt{D}$ system easily yields the Hamiltonian two-component Camassa-Holm equation (CH2) in 2D. In one spatial dimension (1D), the CH2 equation is known to be completely integrable by the isospectral method \cite{CZ2006,F2005}. Its 2D behaviour will be discussed here, briefly but not extensively, because our main goal is the study of a further approximation. The further approximation yields the modified CH2 system (ModCH2). As discovered in \cite{HT2009}, the solution ansatz for the dominant behaviour of ModCH2 is given by the singular momentum map discussed in Theorem \ref{singsolnmommap-thm} equation \eqref{sing-soln-thm} in any number of spatial dimensions. 
Its 1D singular solutions were shown to emerge and dominate the ModCH2 dynamics arising from all smooth, confined, initial conditions discussed in \cite{HOT2009}. As we shall see, the 2D ModCH2 solutions simulated here will not always capture the sharp peaks of the singular solutions. 

Section \ref{sec-3} presents selected computational simulations for several classes of solution behaviour which include wavefront collisions and nonlinear reconnections that are quite reminiscent of the wavefront interaction data shown in the SAR images in figure 1 and figure 2. Our simulations are meant to illuminate the variety of interaction behaviours of the singular momentum map solutions of ModCH2 in 2D. These simulations and our observations of their solution behaviour comprise the primary contribution of the present paper. Additional simulations and videos of their dynamics are provided in the supplementary materials.

\section{Euler-Poincar\'e equations for nonlinearly dispersive  gravity waves}\label{sec-2}

\paragraph{Euler-Poincar\'e formulation of the $1L\sqrt{D}$ equation in 2D}$\,$

The Euler-Poincar\'e formulation of the 2D $1L\sqrt{D}$ equation follows from Hamilton's principle $\delta S = 0$ with $ S = \int_0^T \ell(\bs{u},D)dt$ for the following Lagrangian,%
\footnote{In the $1L\sqrt{D}$ Lagrangian, the term representing kinetic energy of vertical motion is proportional to the Fisher-Rao metric, which appears in probability theory. See e.g., \cite{BR1982} for a fundamental discussion of the Fisher-Rao metric and other generalised information metrics in probability theory. The Fisher-Rao metric is also important in information geometry \cite{V2019}. An equivalent form of the Lagrangian $\ell_{1L\sqrt{D}}$ in terms of spatial gradients is given in \eqref{Lag-1L-B}.}
\begin{align}
\ell_{1L\sqrt{D}}(\bm{u},D)=\frac{1}{2}\int D|\bm{u}|^{2}
+\frac{d^{2}}{3}\left(\frac{\partial}{\partial t}\sqrt{D}\right)^{2}
-g\big(D-b(\bs{x})\big)^{2}\mathrm{d}x\,\mathrm{d}y\,.
\label{Lag-1L-A}
\end{align}
Here, we denote fluid velocity as $\bm{u}$, constant mean layer thickness as $d$, bathymetry as $b(\bs{x})$ with $\bs{x}=(x,y)$, and the total depth as $D$, the last of which satisfies the following advection equation, 
\begin{align}
\frac{\partial D}{\partial t} = -\, {\rm div} (D \bm{u})
\,.
\label{eq-D-cont}
\end{align}
In standard form for fluid dynamics, the  motion equation for $1L\sqrt{D}$ is expressed as
\begin{align}
\frac{\partial\bm{u}}{\partial t}+\bm{u}\cdot\nabla\bm{u}=
-g\nabla\big(D-b(\bs{x})\big)
-\frac{d^{2}}{6}\nabla\left(\frac{1}{\sqrt{D}}\frac{\partial^{2}\sqrt{D}}{\partial t^{2}}\right)
\,.
\label{eq-fluid-1L}
\end{align}

\begin{remark}[An alternative form of the $1L\sqrt{D}$ Lagrangian and energy conservation]$\,$

Upon substituting the continuity equation \eqref{eq-D-cont} into the $1L\sqrt{D}$ Lagrangian in \eqref{Lag-1L-A}, one finds an equivalent Lagrangian, written as the difference of the  kinetic and gravitational potential energies,
\begin{align}
\ell_{1L\sqrt{D}}(\bm{u},D)=\frac{1}{2}\int \bm{u}\cdot Q_{op(D)}\bm{u} 
-g\big(D-b(\bs{x})\big)^{2}\mathrm{d}x\,\mathrm{d}y\,.
\label{Lag-1L-B}
\end{align}
Here, the symmetric, positive-definite operator $Q_{op}(D)$ is defined by its action on the velocity vector, as
\begin{align}
Q_{op(D)}\bm{u} = \Big(D - \frac{d^2}{12}\big(D(\nabla D^{-1}{\rm div})D\big)\Big)\bm{u}
\,.
\label{def-Qop}
\end{align}
After an integration by parts, the conserved sum of the kinetic and potential energies may be  expressed as
\begin{align}
E_{1L\sqrt{D}}(\bm{u},D)=\frac{1}{2}\int D\left[|\bm{u}|^2 + \frac{d^2}{12 D^2}({\rm div}D\bm{u})^2\right]
+ g\big(D-b(\bs{x})\big)^{2}\mathrm{d}x\,\mathrm{d}y\,.
\label{Lag-1L-erg}
\end{align}
The conserved total energy in \eqref{Lag-1L-erg} can be regarded as a metric on the space of smooth functions of vector fields and densities over $\mathbb{R}^2$, $(\bm{u},D)\in\mathfrak{X}(\mathbb{R}^2)\times {\rm Den}(\mathbb{R}^2)$. Hence, one can write the total energy for the $1L\sqrt{D}$ in \eqref{Lag-1L-erg} as a squared norm which defines the following metric on $\mathfrak{X}(\mathbb{R}^2)\times {\rm Den}(\mathbb{R}^2)$,
\begin{align}
E_{1L\sqrt{D}}(\bm{u},D)=: \frac{1}{2}\|(\bm{u},D)\|^2 =: \frac{1}{2} G\big((\bm{u},D)\,,\,(\bm{u},D) \big)
\,.\label{Lag-1L-erg-metric}
\end{align}
The Lie-Poisson Hamiltonian structure of the $1L\sqrt{D}$ model in equations \eqref{eq-D-cont} and  \eqref{eq-fluid-1L} with energy \eqref{Lag-1L-erg} is discussed along with two other models to follow in remark \ref{remark-LPstructure}.
\end{remark}

\begin{remark}[Kelvin circulation theorem for the $1L\sqrt{D}$ motion equation \eqref{eq-fluid-1L}]$\,$

Substituting the Lie-derivative relation,
\begin{align}
\mathcal{L}_{u}(\bm{u}\cdot d\bm{x}) 
= \Big(\bm{u}\cdot\nabla\bm{u} + \nabla \frac{|\bm{u}|^2}{2}\Big)\cdot d\bm{x}
\,,
\label{eq-vector-id}
\end{align}
into the motion equation for $1L\sqrt{D}$ in \eqref{eq-fluid-1L} 
implies the following Kelvin theorem for preservation of circulation,
\begin{align}
\frac{d}{dt}\oint_{c(\bm{u})} \bm{u}\cdot {\rm d}\bm{x} 
= \oint_{c(\bm{u})} \big(\partial_t 
+ \mathcal{L}_{u} \big)(\bm{u}\cdot d\bm{x}) 
= 0\,,
\label{Kel-thm}
\end{align}
for any material loop $c(\bm{u})$ moving with the fluid flow.

\paragraph{The $1L\sqrt{D}$ model admits potential flows.}
The  motion equation for $1L\sqrt{D}$ in \eqref{eq-fluid-1L} and the vector calculus relation in \eqref{eq-vector-id} imply that if ${\rm curl}\bm{u}=0$ initially, then it will remain so. In this case, the corresponding velocity potential $\phi(\bm{x},t)$ for curl-free flows given by $\bm{u}=\nabla\phi$ satisfies a Bernoulli equation given by, 
\begin{align}
\partial_t \phi + \frac12 |\nabla\phi|^2 = -\, g \big(D-b(\bs{x})\big)
-\frac{d^{2}}{6} \left(\frac{1}{\sqrt{D}}\frac{\partial^{2}\sqrt{D}}{\partial t^{2}}\right)
\,,
\label{eq-Bernoulli}
\end{align}
and the continuity equation in \eqref{eq-D-cont} becomes
\begin{align}
\frac{\partial D}{\partial t} = -\, {\rm div} (D \nabla\phi)
\,,
\label{eq-D-cont-phi}
\end{align}
for potential flows $\bm{u}=\nabla\phi$.

\smallskip

\noindent

\end{remark}

\begin{proposition}
The Euler-Poincar\'e equation for the Lagrangian functional $\ell_{1L\sqrt{D}}(\bm{u},D)$ in \eqref{Lag-1L-A} yields the $1L\sqrt{D}$ motion equation \eqref{eq-fluid-1L} in 2D.
\end{proposition}

\begin{proof}
The Euler-Poincar\'e equation for a Lagrangian functional $\ell(\bm{u},D)$ is given by \cite{CHMR1998,HMR1998}
\begin{align}
\partial_t \frac{\delta \ell}{\delta \bm{u}} + (\bm{u}\cdot \nabla) \frac{\delta \ell}{\delta \bm{u}}
+ \frac{\delta \ell}{\delta \bm{u}}({\rm div} \bm{u})
+ \frac{\delta \ell}{\delta u^j} \nabla u^j 
=
D \nabla \frac{\delta \ell}{\delta D}
%
%\partial_t \frac{\delta \ell}{\delta \bm{u}} - \bm{u}\times {\rm curl}\frac{\delta \ell}{\delta \bm{u}}
%+ \nabla \left( \bm{u}\cdot\frac{\delta \ell}{\delta \bm{u}}\right)
%=
%D \nabla \frac{\delta \ell}{\delta D}
\,.
\label{EP-eqn}
\end{align}
The corresponding variational derivatives of the Lagrangian functional $\ell_{1L\sqrt{D}}(\bm{u},D)$ in \eqref{Lag-1L-A} are given by
\begin{align}
\frac{\delta \ell_{1L\sqrt{D}}}{\delta \bm{u}} = D\bm{u} 
\quad\hbox{and}\quad
\frac{\delta \ell_{1L\sqrt{D}}}{\delta D} = \frac12|\bm{u}|^2 - g \big(D-b(\bs{x})\big)
-\frac{d^{2}}{6} \left(\frac{1}{\sqrt{D}}\frac{\partial^{2}\sqrt{D}}{\partial t^{2}}\right)
\,,
\label{Lag-var-1LD}
\end{align}
in which the variations with respect to $\bm{u}$ and $D$ are taken independently. 
Substitution of the variational derivatives in \eqref{Lag-var-1LD} and the continuity equation \eqref{eq-D-cont} into the Euler-Poincar\'e equation \eqref{EP-eqn} yields the  motion equation for $1L\sqrt{D}$ in \eqref{eq-fluid-1L}.
\end{proof}

\paragraph{Deriving the 2-component Camassa-Holm equation (CH2) in 2D}$\,$

The CH2 equation can be immediately derived as a certain approximation of the $1L\sqrt{D}$ equation. To show this derivation directly,  we first introduce an alternative form for the $1L\sqrt{D}$ Lagrangian. 
The alternative form is derived by inserting the advection equation \eqref{eq-D-cont} into formula \eqref{Lag-1L-A} for the $1L\sqrt{D}$ Lagrangian to find the equivalent expression,
\begin{align}
\ell_{1L\sqrt{D}}(\bm{u},D)=\frac{1}{2}\int D \left(|\bm{u}|^{2}
+\frac{d^{2}}{12D^2}\left({\rm div}D\bm{u}\right)^{2}\right)
-g\big(D-b(\bs{x})\big)^{2}\mathrm{d}x\,\mathrm{d}y
\,.
\label{Lag-1L-B-redux}
\end{align}
We now set $D=d$ in the \emph{kinetic energy} terms only, to find 
\begin{align}
\ell_{CH2}(\bm{u},D)=\frac{1}{2}\int d \left(|\bm{u}|^{2}
+\frac{d^{2}}{12}\left({\rm div}\bm{u}\right)^{2}\right)
-g\big(D-b(\bs{x})\big)^{2}\mathrm{d}x\,\mathrm{d}y
\,.
\label{Lag-CH2}
\end{align}

\begin{proposition}
The Euler-Poincar\'e equation for the Lagrangian functional $\ell_{CH2}(\bm{u},D)$ in \eqref{Lag-CH2} yields the CH2 equation with bathymetry $b(\bs{x})$ in 2D, as follows.
\begin{align}
\partial_t \bm{m'} + (\bm{u}\cdot \nabla) \bm{m'}
+ \bm{m'}({\rm div} \bm{u})
+ m'_j \nabla u^j 
=
-gD\nabla\big(D-b(\bs{x})\big)
%-\frac{d^{2}}{6}\nabla\left(\frac{1}{\sqrt{D}}\frac{\partial^{2}\sqrt{D}}{\partial t^{2}}\right)
\,.
\label{eq-CH2}
\end{align}

\end{proposition}

\begin{proof}
The corresponding variational derivatives of the Lagrangian functional $\ell_{CH2}(\bm{u},D)$ in \eqref{Lag-CH2} are given by
\begin{align}
\bm{m'}:=\frac{\delta \ell_{CH2}}{\delta \bm{u}}
= d \Big(\bm{u} - \frac{d^{2}}{12}\nabla {\rm div}\bm{u} \Big)
 =: Q_{op}(d)\bm{u}
\qquad\hbox{and}\qquad
\frac{\delta \ell_{CH2}}{\delta D} = -\, g \big(D-b(\bs{x})\big)
\,.
\label{Lag-var-CH2}
\end{align}
The Euler-Poincar\'e equation \eqref{EP-eqn} for the variational derivatives in \eqref{Lag-var-CH2} yields the CH2 equation  in \eqref{eq-CH2}.
\end{proof}

\begin{remark}
Setting $g=0$ and $d=2$ in the CH2 equation \eqref{eq-CH2} recovers the two-dimensional version of the Camassa-Holm equation derived in \cite{KSD2001}. For a comprehensive survey of the role of the Camassa-Holm equation in the wider context of nonlinear shallow water equations, see \cite{I-K2021}.
\end{remark}

\begin{remark}[Solving for velocity $\bm{u}$ from momentum $\bm{m'}$ with the grad-div operator in \eqref{Lag-var-CH2}]\label{rem CH2-solve}$\,$

In equation \eqref{Lag-var-CH2} we see that ${\rm curl}(\bm{m'}/d) ={\rm curl} \bm{u}$ and ${\rm div} (\bm{m'}/d) = (1-\frac{d^2}{12}\Delta){\rm div} \bm{u}$. To take a time step in equation \eqref{eq-CH2}we need to determine $\bm{u}$ from $(\bm{m'}/d)$. One could formally write $G_{Q_{op}(d)}*(\bm{m'}/d) = \bm{u}$ in which one implicitly defines the Green function $G_{Q_{op}(d)}$ for the operator $Q_{op}(d)$ in equation \eqref{Lag-var-CH2}. However, it is also useful to verify that the solution for the velocity $\bm{u}$ from the CH2 momentum $(\bm{m'}/d)$ can be implemented directly, without solving for the Green function explicitly.

The velocity $\bm{u}$ can be obtained from momentum $\bm{m'}$ via their linear operator relation for velocity $\bm{u}$ in equation \eqref{Lag-var-CH2}. For this, one begins with the Hodge decomposition for the velocity. Namely,
\begin{align}
\bm{u} = {\rm curl}\bm{A} + \nabla \phi
\,.\label{Hodge-CH2}
\end{align}
Here, the vector potential $\bm{A}$ is divergence-free ${\rm div}\bm{A}=0$, has zero  mean $\int_{\cal D} \bm{A} d^2x =0$, and satisfies Neumann boundary conditions, $\partial_n\bm{A}|_{\partial \cal D}=0$. The scalar potential $\phi$ vanishes at the boundary. With these conditions, the vector and scalar potentials each satisfy Poisson equations,
\begin{align}
{\rm div}\bm{u}=\Delta\phi
\,, \quad\hbox{and}\quad
{\rm curl}\bm{u}=-\Delta\bm{A}
\,.
\label{PoissonEqns-CH2}
\end{align}

Taking the curl and divergence of the defining relation in \eqref{Lag-var-CH2} for the momentum $\bm{m'}$ in terms of the velocity $\bm{u}$ then yields the inversion formulas for the velocity potentials,
\begin{align}
{\rm div}(\bm{m'}/d) = (1-\alpha^2\Delta){\rm div}\bm{u} = (1-\alpha^2\Delta)\Delta\phi
\,, \quad\hbox{and}\quad
{\rm curl}(\bm{m'}/d) = {\rm curl}\bm{u}=-\Delta\bm{A}
\,.\label{m-u-CH2}
\end{align}
Inverting the relations in \eqref{m-u-CH2} for the vector and scalar velocity potentials $\bm{A}$ and $\phi$ then yields the velocity via the Hodge decomposition in \eqref{Hodge-CH2}.
\end{remark}

\begin{remark}[Kelvin theorem and conservation laws for the CH2 motion equation \eqref{eq-CH2}]\label{KelThm-conslawsCH2}$\,$

A geometric way of writing the Euler-Poincar\'e equation for any Lagrangian functional $\ell(\bm{u},D)$ in $n$ dimensions arises by regarding the fluid velocity as a vector field, denoted $u$, the depth as an $n$-form, denoted ${\sf D}$, and its dual momentum density as a 1-form density, denoted $m:={\delta \ell}/{\delta u}$. In coordinates, this is \cite{CHMR1998,HMR1998}
\begin{align}
u :=\bm{u}\cdot\nabla \in \mathfrak{X}(\mathbb{R}^n)
\,,\quad
{\sf D}=D\,{\rm d}^nx\in {\rm Den}(\mathbb{R}^n)
\quad\hbox{and}\quad
\frac{\delta \ell}{\delta u} 
:= \frac{\delta \ell}{\delta \bm{u}} \cdot {\rm d} \bm{x}\otimes {\rm d}^nx \in \mathfrak{X}^*(\mathbb{R}^n)
\,.
\label{u&m-def}
\end{align}
For the Euler-Poincar\'e variational principle, one also assumes a natural $L^2$ pairing, $\langle\,\cdot\,,\,\cdot\,\rangle$, so that 
\begin{align}
0 = \delta S = \int \left\langle \frac{\delta \ell}{\delta u}\,,\, \delta u \right\rangle
+ \left\langle \frac{\delta \ell}{\delta {\sf D}}\,,\, \delta {\sf D} \right\rangle
\,{\rm d}t
=
0
\,.
\label{EP-pair}
\end{align}
In this framework, the Euler-Poincar\'e equation \eqref{EP-eqn} for a Lagrangian functional $\ell(u,{\sf D})$ and the auxiliary equation for the advection of density $D$ are given by \cite{CHMR1998,HMR1998}
\begin{align}
\big({\partial_t} + \mathcal{L}_u\big) \frac{\delta \ell}{\delta u}
=
{\sf D}\, {\rm d} \frac{\delta \ell}{\delta {\sf D}}
\quad\hbox{and}\quad
\big({\partial_t} + \mathcal{L}_u\big) {\sf D} = 0
\,.
\label{EP-eqn-Lie}
\end{align}
where ${\sf D}$ is a density and ${\delta \ell}/{\delta {\sf D}}$ is a scalar function, according to our $L^2$ pairing in \eqref{EP-pair}.

%{\bf The presence of the operator $\partial_t + \mathcal{L}_u$ means \dots pull-back push-forward }

The Kelvin circulation theorem in this framework is then proved, as follows,
\begin{align}
\frac{d }{d t}\oint_{c(u)} \frac{1}{\sf D}\frac{\delta \ell}{\delta u} 
=
\oint_{c(u)} \big({\partial_t} + \mathcal{L}_u\big) \Big(\frac{1}{\sf D}\frac{\delta \ell}{\delta u}\Big)
=
\oint_{c(u)} {\rm d} \frac{\delta \ell}{\delta {\sf D}} = 0
\,.
\label{Kel-Lie}
\end{align}
In particular, according to \eqref{eq-CH2} and \eqref{Lag-var-CH2}, the Kelvin circulation theorem for CH2 with Lagrangian $\ell_{CH2}$ in \eqref{Lag-CH2} is given by 
\begin{align}
\frac{d }{d t}\oint_{c(u)} \frac{1}{\sf D}\frac{\delta \ell_{CH2}}{\delta u} 
=
\frac{d }{d t}\oint_{c(u)} \frac{d}{D}\Big(\bm{u} - \frac{d^{2}}{12}\nabla {\rm div}\bm{u} \Big)\cdot {\rm d}\bm{x}
= 0
\,.
\label{Kel-CH2}
\end{align}
For CH2 in 2D, applying the Stokes theorem to the Kelvin theorem \eqref{Kel-CH2} implies conservation along flow trajectories of a \emph{potential vorticity}. Namely,
\begin{align}
\partial_t \sigma +  \bm{u} \cdot\nabla \sigma = 0 
\,,\quad\hbox{where}\quad
\sigma := \bm{\hat{z}}\cdot {\rm curl}\left(
\frac{d}{D}\Big(\bm{u} - \frac{d^{2}}{12}\nabla {\rm div}\bm{u} \Big)\right)
\,.
\label{Stokes-CH2-2D}
\end{align}
This means, in particular, that if $\sigma$ vanishes initially, it will continue to do so. 

Moreover, equation \eqref{Stokes-CH2-2D} and the continuity equation in \eqref{eq-D-cont} imply preservation of the integral quantities (enstrophies) given by
\begin{align}
C_\Phi:=\int \Phi(\sigma)\,D{\rm d}x{\rm d}y
\,,
\label{Casimirs-CH2}
\end{align}
for any differentiable function $\Phi$.

For CH2 in 3D, applying the Stokes theorem to the Kelvin circulation conservation law for the CH2 model in {Kel-CH2} implies the advection of a potential-vorticity vector field, $\bm{\sigma}$, given in components by
\begin{align}
\partial_t \bm{\sigma} +  \bm{u} \cdot\nabla \bm{\sigma} - \bm{\sigma}\cdot\nabla\bm{u}  = 0 
\,,\quad\hbox{where}\quad
\bm{\sigma} := D^{-1} {\rm curl}\bm{v} 
\,,\quad\hbox{with}\quad
\bm{v} = D^{-1}\bm{m'} = 
\frac{d}{D}\Big(\bm{u} - \frac{d^{2}}{12}\nabla {\rm div}\bm{u} \Big)
\,.
\label{Stokes-CH2-3D}
\end{align}
In 3D, the CH2 equation \eqref{Stokes-CH2-2D} and the continuity equation in \eqref{eq-D-cont} imply preservation of the integral quantity (helicity) given by
\begin{align}
\Lambda :=\int_{\mathcal{B}_t } \bm{v}\cdot{\rm curl}\bm{v} \,{\rm d}^3x 
= \int_{\mathcal{B}_t } \bm{v}\cdot{\rm d}\bm{x}\wedge {\rm d} (\bm{v}\cdot{\rm d}\bm{x})
\,.
\label{helicity-CH2-3D}
\end{align}
The helicity integral in \eqref{helicity-CH2-3D} is taken over any volume (blob)  ${\cal B}_t = \phi_t{\cal B}_0$ of fluid moving with the flow, $\phi_t$, with outward normal boundary condition ${\rm curl}\bm{v}\cdot \bm{\hat{n}}=0 $ on the surface $\partial {\cal B}$.
\end{remark}

\paragraph{Deriving the Modified 2-component Camassa-Holm equation (ModCH2) in 2D and 3D}$\,$

\noindent
To derive the ModCH2 model equations, we modify the \emph{potential energy} terms in the Lagrangian \eqref{Lag-CH2} for the CH2 model, as follows
\begin{align}
\begin{split}
\ell_{ModCH2}(\bm{u},D)&=\frac{1}{2}\int d \left(|\bm{u}|^{2}
+\frac{d^{2}}{12}\left({\rm div}\bm{u}\right)^{2}\right)
\\& \qquad -g\left[\big(D-b(\bs{x})\big)G_{Q_{op}(d)}*\big(D-b(\bs{x})\big) 
\right]\mathrm{d}^nx
\,,
\end{split}
\label{Lag-ModCH2}
\end{align}
where convolution with the Green function $G_{Q_{op}(d)}$ acts as a smoothing operator in the potential energy term in the ModCH2 Lagrangian. The $({\rm div}\bm{u})^2$ term in \eqref{Lag-ModCH2} effectively replaces the vertical kinetic energy by the divergence of the horizontal velocity.

\begin{proposition}
The Euler-Poincar\'e equation for the Lagrangian functional $\ell_{ModCH2}(\bm{u},D)$ in \eqref{Lag-ModCH2} yields the following \emph{modified} CH2 equation,
\begin{align}
\partial_t \bm{m} + (\bm{u}\cdot \nabla) \bm{m}
+ \bm{m}({\rm div} \bm{u})
+ m_j \nabla u^j 
=
-gD\nabla G_{Q_{op}(d)}*\big({D}-{b}(\bs{x})\big)
=:
-gD\nabla\big(\overline{D}-\overline{b}(\bs{x})\big).
\label{eq-ModCH2}
\end{align}

\end{proposition}

\begin{proof}
The corresponding variational derivatives of the Lagrangian functional $\ell_{ModCH2}(\bm{u},D)$ in equation  \eqref{Lag-ModCH2} are given by
\begin{align}
\bm{m} := \frac{\delta \ell_{ModCH2}}{\delta \bm{u}} 
= d \Big(\bm{u} - \alpha^2\nabla {\rm div}\bm{u} \Big) 
\quad\hbox{and}\quad
\frac{\delta \ell_{ModCH2}}{\delta D} = -g \big(\overline{D}-\overline{b}(\bs{x})\big)
\,.
\label{Lag-var-ModCH2}
\end{align}
The Euler-Poincar\'e equation \eqref{EP-eqn} for the variational derivatives in \eqref{Lag-var-ModCH2} yields the ModCH2 equation in \eqref{eq-ModCH2}. 
\end{proof}

In summary, the velocity $\bm{u}$ in the ModCH2 equation \eqref{eq-ModCH2} is obtained from momentum $\bm{m}$ at each time step by inverting the grad-div Helmholtz operator in \eqref{Lag-var-ModCH2} as explained in remark \ref{rem CH2-solve}. This procedure is valid in both 2D and 3D.  

\begin{remark}[Kelvin theorem, conservation laws and an additional property for ModCH2]\label{KelThmConsMomap-ModCH2}$\,$

\noindent
In combination with the continuity equation for the total depth as $D$ in \eqref{eq-D-cont}, the Kelvin theorem and conservation laws for ModCH2 may be obtained as analogues of those for CH2 in both 2D and 3D in remark \ref{KelThm-conslawsCH2}. 
However, ModCH2 was introduced in \cite{HOT2009} to provide an additional structural feature which goes beyond the CH2 equation. Namely, ModCh2 is both a geodesic equation and an Euler-Poincar\'e equation. In the next section, we will discuss the implications of these dual properties.  
\end{remark}

In addition to possessing the same Kelvin theorem and all of the corresponding conservation laws for CH2 discussed in remark \ref{KelThm-conslawsCH2} which accompany its derivation as an Euler-Poincar\'e equation, the Lie-Poisson Hamiltonian formulation of the ModCH2 equation places it into a class of equations which admit singular momentum map solutions in any number of dimensions. This is the subject of the next section.

\subsection{Singular momentum map solutions for Modified CH2 (ModCH2)}\label{ssec: singmomap}

The purpose of this section is to explain how the dual properties of ModCH2 in being both a geodesic equation and an Euler-Poincar\'e equation endow it with singular momentum map solutions in any number of dimensions.
That is, ModCH2 admits singular solutions that are represented as a sum over Dirac deltas supported on curves in the plane, or surfaces in three dimensions, which are advected by the flow of the currents which they themselves induce throughout the rest of the domain. 

Specifically, the singular solutions are given in Theorem \ref{singsolnmommap-thm} by
\begin{equation}
\big({\bf m},D\big)=\sum_{i=1}^N\int\!\big({\bf P}_i(s,t),w_i(s)\big)\,\delta\!\left({\bf
x-Q}_i(s,t)\right)\,{\rm d}^ks
\,,
\label{SDsingsoln}
\end{equation}
where $s$ is a coordinate on a submanifold $S$ of $\mathbb{R}^n$, exactly as in the case of EPDiff. For $\mathbb{R}^2$, the case dim$\,S=1$ yields fluid variables supported on filaments moving under the action of the diffeomorphisms, while for $\mathbb{R}^3$ dim$\,S=2$ yields fluid variables supported on moving surfaces. 
The geometric setting of the peakon solutions of the Camassa-Holm equation and its extension to pulson solutions of EPDiff was established in \cite{HMR1998}. Following the reasoning in \cite{HM2005,HT2009},
one may interpret ${\bf Q}_i$ in \eqref{SDsingsoln} as a smooth embedding in Emb$(S,\mathbb{R}^n)$ and $P_i={\bf P}_i\cdot{\rm d}{\bf Q}_i$ (no sum) as the canonical 1-form on the cotangent bundle $T^*$Emb$(S,\mathbb{R}^n)$ for the $i$-th smooth embedding.

In a sense, the singular ModCH2 wave-currents are analogues for nonlinear wave dynamics of point vortices in 2D and vortex lines in 3D for Euler fluid dynamics. However, unlike vortices and vortex lines in 3D for Euler fluids, the singular ModCH2 wave-currents can emerge spontaneously from smooth, spatially confined initial conditions, while the point vortices and vortex lines do not emerge spontaneously in Euler fluid dynamics. 

\begin{remark}[Shared Lie-Poisson Hamiltonian structure]\label{remark-LPstructure}
As we have seen, all of the models $1L\sqrt{D}$, CH2 and ModCH2 yield semiditect-product Euler-Poincar\'e equations in the class EP(Diff$\,\circledS\,\mathcal{F}$) in equation \eqref{EP-eqn}. Here, $\mathcal{F}$ comprises the smooth scalar functions of the densities ${\sf D}=D\,{\rm d}^nx\in {\rm Den}(\mathbb{R}^n)$ and $\circledS$ denotes the semidirect-product action \cite{CHMR1998,HMR1998}.

In $n$ dimensions, the corresponding Lie-Poisson Hamiltonian equations can be obtained from the Legendre transformation, 
\begin{align}
h(m,{\sf D}):= \langle m,u\rangle - \ell(u,{\sf D})
\,.\label{Leg-xform}
\end{align}
The variational derivatives of the Hamiltonian are given by
\begin{align}
\delta h(m,{\sf D}) = \big\langle \delta m,u\big\rangle 
+ \Big\langle m - \frac{\delta \ell}{\delta u},\delta u \Big\rangle
+ \Big\langle -\,\frac{\delta \ell}{\delta {\sf D}},\delta {\sf D} \Big\rangle
\,.\end{align}
Under the Legendre transformation \eqref{Leg-xform}, the semidirect-product Lie-Poisson Hamiltonian equations corresponding to the Euler-Poincar\'e equations in \eqref{EP-eqn} can be written in three-dimensional matrix component form, as \cite{CHMR1998,HMR1998}
\begin{align}
\frac{\partial }{\partial t } 
\begin{bmatrix}
 m_i \\
 D
\end{bmatrix}
=
-
\begin{bmatrix}
\partial_j m_i + m_i \partial_j & D \partial_i
\\
 \partial_j D  & 0 
\end{bmatrix}
\begin{bmatrix}
\frac{\delta h}{\delta m_j} = u^j \\
 \frac{\delta h}{\delta D} = - \frac{\delta \ell}{\delta D}
\end{bmatrix}
\,.
\label{LP-eqn}
\end{align}
In \eqref{LP-eqn}, one sums over repeated spatial component indices, $i,j=1,2,3$, for each of the Lagrangians $\ell_{1L\sqrt{D}}$, $\ell_{CH2}$, and $\ell_{ModCH2}$, and all three motion equations share the continuity equation for the total depth, $D$, 
\begin{align}
\frac{\partial D}{\partial t} = -\, {\rm div} (D \bm{u})
\,.
\label{eq-D-cont-redux}
\end{align}

When the Lie-Poisson matrix form \eqref{LP-eqn} is extended to $n$ dimensions, the $1L\sqrt{D}$ equations 
describe geodesic motion with respect to the following metric Hamiltonian
\begin{align}
h_{1L\sqrt{D}}(\bm{u},D)=\frac{1}{2}\int \bm{m}\cdot G_{Q_{op}(D)}*\bm{m} 
+g\big(D-b(\bs{x})\big)^{2}\mathrm{d}^nx
\,,
\label{Ham-1L-B}
\end{align}
in which $G_{Q_{op}(D)}$ is the Green function for the symmetric operator $Q_{op}(D)$ in equation \eqref{def-Qop}. That is, 
\begin{align}
G_{Q_{op}(D)}*\bm{m}=\bm{u}
\label{def-Qop(D)}
\end{align}
is the velocity vector for the $1L\sqrt{D}$ model. 

Likewise, the ModCH2 equations describe geodesic motion with respect to the metric Hamiltonian obtained by replacing $G_{Q_{op}(D)}$ by $G_{Q_{op}(d)}$ in equation \eqref{Ham-1L-B}. The ModCH model also has the special feature that its Hamiltonian 
\begin{align}
\begin{split}
h_{ModCH2}(\bm{m},D)&=\frac{1}{2}\int \bm{m}\cdot G_{Q_{op}(d)}*\bm{m}\
+\
g \big(D-b(\bs{x})\big)G_{Q_{op}(d)}*\big(D-b(\bs{x})\big) \,\mathrm{d}^nx
\,,
\end{split}
\label{Ham-ModCH2}
\end{align}
lies in the following class of general metrics (Green functions),
\begin{equation}\label{EPGosH-Ham}
H({\bf m},D)=
\frac12\iint\!{\bf m}({\bf x},t)\cdot \,G_1({\bf x-x}')\,{\bf m}({\bf x}',t)\,{\rm d}^n{\bf x}\,{\rm d}^n{\bf x}'
+
\frac12\iint\!D({\bf x},t)\,G_2({\bf x-x}')\,D({\bf x}',t)\,{\rm d}^n{\bf x}\,{\rm d}^n{\bf x}'
\,.\end{equation}
Importantly for the remainder of the present work, the class of Hamiltonians in \eqref{EPGosH-Ham} admits emergent singular solutions supported on advected embedded spaces.
\end{remark}
%%%%%%%%%%%%%%%%%%%%%%%%%%%%%%%%%%%%%%%%%%%%%%%%%%%%%%%%%%%%%%%%%%%%

In preparation for displaying the computational simulations of the singular solution behaviour for ModCH2, we write the equations in dimension-free form. 

\begin{remark}[Dimension-free form of ModCH2 Lagrangian, $\ell_{ModCH2}$]\label{non-dim-scales}
The Lagrangian for ModCH2 in \eqref{Lag-ModCH2} may be cast into dimension-free form by introducing natural units for horizontal length, $[L]$, horizontal velocity, $[U]$, and time, $[T]=[L]/[U]$, as well as spatially mean vertical depth, denoted $d=\langle D \rangle$, and spatially mean vertical wave elevation, $[\zeta]=\langle D-b(\bs{x})\rangle=d-\langle b\rangle$. In terms of these units, one may define the following two dimension-free parameters: aspect ratio, $\sigma = [d]/[L]$ and elevation Froude number squared, $Fr^2=[U]^2/([g][\zeta])$. 
In addition, the dimension-free form of the symmetric operator $Q_{op}(\sigma)$ is redefined with $\alpha^2:=\sigma^2/12$ as
\begin{align}
Q_{op(\alpha)}\bm{u} := \Big(1 - \alpha^2\big(\nabla {\rm div}\big)\Big)\bm{u}
\,.
\label{def-Qop-ndim}
\end{align}
Consequently, the dimension-free form of the Lagrangian for ModCH2 is given by
\begin{align}
\begin{split}
\ell_{ModCH2}(\bm{u},D)&=\frac{1}{2}\int  \left(|\bm{u}|^{2}
+\alpha^2\left({\rm div}\bm{u}\right)^{2}\right)
\\& \qquad - Fr^{-2}\left[\big(D-b(\bs{x})\big)G_{Q_{op}(\alpha)}*\big(D-b(\bs{x})\big) 
\right]\mathrm{d}^nx
\,,
\end{split}
\label{Lag-ModCH2-ndim}
\end{align}
The constants $\sigma^2\ll1$ and $Fr^{-2}=O(1)$ here are, respectively, the squares of the aspect ratio and the Froude number, which have been obtained in making the expression dimension-free. The final dimension-free number to be defined in the simulations will be the ratio of widths obtained by dividing the width of the initial condition by the filter width, or interaction range, $\alpha=d/\sqrt{12}$ in $Q_{op(\alpha)}$, as defined in \eqref{def-Qop-ndim}.

\end{remark}

\begin{proposition}
The Euler-Poincar\'e equation for the dimension-free Lagrangian functional $\ell_{ModCH2}(\bm{u},D)$ in \eqref{Lag-ModCH2-ndim} yields the dimension-free form of the ModCH2 equation, as follows.
\begin{align}
\partial_t \bm{m} + (\bm{u}\cdot \nabla) \bm{m}
+ \bm{m}({\rm div} \bm{u})
+ m_j \nabla u^j 
=
-Fr^{-2}D\nabla G_{Q_{op}(\alpha)}*\big({D}-{b}(\bs{x})\big)
=:
-Fr^{-2}D\nabla\big(\overline{D}-\overline{b}(\bs{x})\big).
\label{eq-ModCH2-ndim}
\end{align}

\end{proposition} 

\begin{proof}
The corresponding variational derivatives of the Lagrangian functional $\ell_{ModCH2}(\bm{u},D)$ in equation \eqref{Lag-ModCH2-ndim} are given by
\begin{align}
\bm{m} := \frac{\delta \ell_{ModCH2}}{\delta \bm{u}} 
= \Big(\bm{u} - \alpha^2\nabla {\rm div}\bm{u} \Big) 
\quad\hbox{and}\quad
\frac{\delta \ell_{ModCH2}}{\delta D} = -Fr^{-2} \big(\overline{D}-\overline{b}(\bs{x})\big)
\,.
\label{Lag-var-ModCH2-ndim}
\end{align}
The Euler-Poincar\'e equation \eqref{EP-eqn} for the variational derivatives in \eqref{Lag-var-ModCH2-ndim} yields the ModCH2 equation in \eqref{eq-ModCH2-ndim}. 
\end{proof}

\paragraph{Singular ModCH2 solutions.}
In ideal fluid dynamics, various conservation laws are expressed on advected embedded spaces such as loops and surfaces, as we have seen for example, in the case of CH2 in 2D in the previous discussion. As explained in \cite{HT2009}, ModCH2 dynamics is dominated by the emergence of weak solutions supported on advected embedded spaces. These emergent weak solutions for ModCH2 on embedded spaces in any dimension define the momentum map for the left action of the diffeomorphisms on functions taking values on the semidirect-product Lie algebra  $\mathfrak{X}(\mathbb{R}^n)\circledS V(\mathbb{R}^n)$. (This is the Lie algebra of vector fields acting on functions which take values on vector spaces $V$ over $\mathbb{R}^n$, \cite{G-BV2012,HM2005,HT2009}.%
\footnote{
In contrast, point vortices in Euler flow define a \emph{symplectic} momentum map \cite{MW1983} which also generalises to higher-order derivatives, \cite{HJ2009,CHJM2014,CEHJM2016}.} 

The singular momentum map we shall discuss here arises as part of a dual pair.%
\footnote{See \cite{W1983} for the definitive discussion of dual pairs.} 
The rigid body provides a familiar example of a dual pair. In the rigid body, the two legs of the dual pair correspond to the cotangent lift momentum maps for right and left actions, respectively. The dual pair for Euler fluids implies (from right-invariance) that the momentum map $J_R$ is conserved. For Euler fluids, the conservation of the right momentum map $J_R$ is equivalent to Kelvin's circulation theorem. For Euler fluids, the left momentum map $J_L$ maps Hamilton’s canonical equations on $T^*(SDiff)$ to their reduced Lie-Poisson form and at the same time implies that the solutions on $T^*(SDiff)$ can be defined on embedded subspaces of the domain of flow which are pushed forward by the left action of $SDiff$ \cite{HM2005}. These results for ideal incompressible Euler fluids were generalised to semidirect-product left action of $Diff$ on embedded subspaces of the domain of flow for ideal compressible fluids in \cite{HT2009}. For the fundamental proofs that these maps satisfy the technical conditions required for verifying them as dual pairs, see \cite{G-BV2012}.

In summary, for the semidirect-product case of EP(Diff$\,\circledS\,\mathcal{F}$), the weights $w_i$ for $i=1,\dots,N$ in \eqref{SDsingsoln} are considered as maps $w_i:S\to\mathbb{R}^*$. That is, the weights $w_i$ are  distributions on $S$, so that $w_i\in{\rm Den}(S)$, where ${\rm Den}:=\mathcal{F}^*$. In particular, considering the triple
\[
({\bf Q}_i,{\bf P}_i, w_i)\ \, \text{\large $\in$} \ \,
T^*{\rm Emb}(S,\mathbb{R}^n)\,\times\,{\rm Den}(S)
\,,
\]
leads to the following solution momentum map introduced in \cite{HT2009}.

\begin{theorem}[Singular solution momentum map \cite{HT2009}]
\label{singsolnmommap-thm}$\,$

\noindent
The singular solutions of the semidirect-product Lie-Poisson equations in \eqref{LP-eqn} for $\ell=\ell_{ModCH2}$ in \eqref{Lag-ModCH2} are given by
\begin{align}
\big({\bf m},D\big)=\sum_{i=1}^N\int\!\big({\bf P}_i(s,t),w_i(s)\big)\,\delta\!\left({\bf
x-Q}_i(s,t)\right)\,{\rm d}^ks
\,.
\label{sing-soln-thm}
\end{align}
The expressions for $({\bf m},D)\in\mathfrak{X}^*(\mathbb{R}^n)\ \text{\large $\circledS$}\ {\rm Den}(\mathbb{R}^n)$ in \eqref{sing-soln-thm} identify a momentum map
\begin{align}
{\bf J}:\underset{i=1}{\overset{N}{\text{\LARGE$\times$}}}\Big(T^*{\rm Emb}(S,\mathbb{R}^n)\,\times\,{\rm Den}(S)\Big)
\,\rightarrow\,
\mathfrak{X}^*(\mathbb{R}^n)\ \text{\large $\circledS$}\ {\rm Den}(\mathbb{R}^n)
\,.
\end{align}
After substituting the formulas in \eqref{sing-soln-thm} for the singular solutions into the Hamiltonian $H({\bf m},D)$ in equation \eqref{EPGosH-Ham}, the dynamical equations for the variables 
$({\bf Q}_i,{\bf P}_i, w_i)$ are given by the integral expressions
\begin{align}
\begin{split}
\partial_t{{\bf Q}_i(s,t)} 
=&\ \sum_j\int\!{\bf P}_j(s',t)\, G_1({\bf
Q}_i(s,t)-{\bf Q}_j(s',t))\ {\rm d}^ks' \,,
\\
\partial_t{{\bf P}_i(s,t)} =&\ -\sum_j\int\! {\bf P}_i(s,t)\cdot{\bf
P}_j(s',t)\,\text{\large$\nabla$}_{\!{\bf Q}_i}G_1({\bf
Q}_i(s,t)-{\bf Q}_j(s',t))\ {\rm d}^ks'
\\
&\qquad - \sum_j\int\! w_i(s)\,w_j(s')\, \text{\large$\nabla$}_{\!{\bf
Q}_i}G_2({\bf Q}_i(s,t)-{\bf Q}_j(s',t))\ {\rm d}^ks' \,,
\end{split}
\label{singsoldynamics}
\end{align}
with $\partial_t{w}_i(s)=0$, for all values of $i$.

\end{theorem}

The considerations discussed in \cite{HT2009,G-BV2012} derive the above singular momentum map as the left-invariant leg of a defined dual pair. However, these considerations will not be reviewed here. Instead, the next section will start a series of illustrations by numerical simulations of the dynamical behaviour of the solutions of the ModCH2 equations \eqref{eq-ModCH2} in 2D with periodic boundary conditions.

\section{Solution behaviour -- simulations}\label{sec-3}

%\begin{outline}[enumerate]
%\0 Background 
% \1 Euler-Poincar\'e derivation
%0 Solution behaviour -- simulated solutions of the velocity magnitude $|\bu|$ and elevation $\overline{D}$ depend on the initial conditions. In particular, the shapes of the individual wave profiles in an emerging wavetrain depend on the initial conditions.
%\end{outline}
%\bigskip

\subsection{Background -- Euler-Poincar\'e and Lie-Poisson derivations}

\noindent
This section reports computational simulations of the interaction dynamics of wave fronts. Before embarking on our report of these computational simulations, let us place them into the context of the previous literature, which is based on approximating the Lagrangian in Hamilton's principle for fluid dynamics. Such approximations have been designed before to preserve the transport and topological properties of variational principles \cite{CHMR1998,Holm2001,Holm2015,FHT2001,FHT2002,HMR1998,HOT2009,HS2013}. 

Specifically, this section reports simulations in which the Lagrangian functional $\ell_{ModCH2}(\bm{u},D)$ in equation \eqref{Lag-ModCH2-ndim} has been augmented to complete the $H^1$ norm in the kinetic energy of the dimension-free form of the Lagrangian. Namely, we modify the Lagrangian $\ell_{ModCH2}$ in \eqref{Lag-ModCH2-ndim} to include the full H1 norm by writing,  
\begin{align}
\begin{split}
\ell_{H1ModCH2}(\bm{u},D)&=\frac{1}{2}\int  \left(|\bm{u}|^{2}
+\alpha^2\left|\nabla\bm{u}\right|^{2}\right)
\\& \qquad - Fr^{-2}\left[\big(D-b(\bs{x})\big)G_{Q^{H1}_{op}(\alpha)}*\big(D-b(\bs{x})\big) 
\right]\mathrm{d}^nx
\,.
\end{split}
\label{Lag-H1ModCH2-ndim}
\end{align}
The corresponding dimension-free form of the symmetric operator $Q^{H1}_{op}(\sigma)$ for the H1ModCH2 equation is defined as, cf. $Q_{op(\alpha)}\bm{u}$ in equation \eqref{def-Qop-ndim},
\begin{align}
Q^{H1}_{op(\alpha)}\bm{u} := \Big(1 - \alpha^2\big({\rm div}\nabla \big)\Big)\bm{u}
\,.
\label{def-Qop-ndim}
\end{align}

\begin{proposition}
The Euler-Poincar\'e equation for the dimension-free Lagrangian functional $\ell_{H1ModCH2}(\bm{u},D)$ in \eqref{Lag-H1ModCH2-ndim} yields the dimension-free form of the H1ModCH2 equation, as follows.
\begin{align}
\begin{split}
\partial_t \bm{\widetilde m} + (\bm{u}\cdot \nabla) \bm{\widetilde m}
+ \bm{\widetilde m}({\rm div} \bm{u})
+ {\widetilde m}_j \nabla u^j 
&=
-Fr^{-2}D\nabla G_{Q^{H1}_{op}(\alpha)}*\big({D}-{b}(\bs{x})\big)
\\&=:
-Fr^{-2}D\nabla\big(\widetilde{D}-\widetilde{b}(\bs{x})\big),
\end{split}
\label{eq-H1ModCH2-ndim}
\end{align}
where $\widetilde{(\,\cdot\,)}:=G_{Q^{H1}_{op}(\alpha)}*(\,\cdot\,)$ 
and ${\widetilde m}_j:={\delta \ell_{H1ModCH2}}/{\delta {u^j}}=Q^{H1}_{op(\alpha)}u_j$.
\end{proposition}

\begin{proof}
The corresponding variational derivatives of the Lagrangian functional $\ell_{H1ModCH2}(\bm{u},D)$ in equation \eqref{Lag-H1ModCH2-ndim} in $n$ dimensions are given by
\begin{align}
\bm{\widetilde m} := \frac{\delta \ell_{H1ModCH2}}{\delta \bm{u}} 
= \Big(\bm{u} - \alpha^2 \big({\rm div}\nabla \big) \bm{u} \Big) 
\quad\hbox{and}\quad
\frac{\delta \ell_{H1ModCH2}}{\delta {\widetilde D}} 
= -Fr^{-2} \big(\widetilde{D}-\widetilde{b}(\bs{x})\big)
\,.
\label{Lag-var-H1ModCH2-ndim}
\end{align}
The Euler-Poincar\'e equation \eqref{EP-eqn} for the variational derivatives in \eqref{Lag-var-ModCH2-ndim} yields the H1ModCH2 equation in \eqref{eq-H1ModCH2-ndim}. After a Legendre transformation, the semidirect-product Lie-Poisson Hamiltonian equations in \eqref{LP-eqn} also yield the H1ModCH2 equation in \eqref{eq-H1ModCH2-ndim}, as well as the continuity equation for the depth, $D$. 
\end{proof}

\begin{remark}
The Hamiltonian for the H1ModCH2 equation in \eqref{eq-H1ModCH2-ndim} is given by 
\begin{align}
\begin{split}
h_{H1ModCH2}(\bm{\widetilde m},D)&=\frac{1}{2}\int \bm{\widetilde m}\cdot G_{Q^{H1}_{op}(d)}*\bm{\widetilde m}\
+\
g \big(D-b(\bs{x})\big)G_{Q^{H1}_{op}(d)}*\big(D-b(\bs{x})\big) \,\mathrm{d}^nx
\,,
\end{split}
\label{Ham-H1ModCH2}
\end{align}
This Hamiltonian also lies in the class the class of Hamiltonians in \eqref{EPGosH-Ham}. Consequently, the H1ModCH2 equation will admit emergent singular solutions supported on advected embedded spaces which dominate their asymptotic behaviour.

\end{remark}

\begin{remark}[Interaction dynamics of singular wave fronts with both kinetic and potential energy]$\,$

\noindent 
When gravitational potential energy is neglected, then $Fr^{-2}\to0$ in equations \eqref{eq-H1ModCH2-ndim} and \eqref{Lag-var-H1ModCH2-ndim}. The dimension-free form of the ModCH2 equation \eqref{eq-H1ModCH2-ndim} in $n$-dimensions then reduces to the $n$D Camassa-Holm equation, which was introduced in \cite{HMR1998} and studied numerically in 2D and 3D in \cite{HS2013}. For divergence-free flows, the $n$D Camassa-Holm equation is also known as the Euler-$\alpha$ model in a class of other $\alpha$-fluid models \cite{HMR1998PRL}, and it was the source of the Lagrangian-Averaged Navier-Stokes $\alpha$ model (LANS-$\alpha$ model) of divergence-free turbulence in \cite{CFHOTW1998,CFHOTW1999,FHT2001,FHT2002}. In the remainder of the present paper, we will present computational simulations of equation \eqref{eq-H1ModCH2-ndim} in 2D which include the effects of potential energy as well as the vorticity in the interaction of singular wave fronts. Consequently, the results we will obtain may be compared with computational simulations of the Camassa-Holm equation in 2D and 3D of \cite{HS2013}, in order to see the differences made in the solution behaviour due to the presence of gravitational potential energy. The 1D version of these comparisons have already been made for solutions of both the ModCH2 and H1ModCH2 equations in \cite{HOT2009}. In 1D the same run accomplishes the comparisons of the Camassa-Holm solutions with those of both ModCH2 and H1ModCH2, because in 1D the operators ${\rm div}\nabla$ and $\nabla{\rm div}$ are the same. The work here presents solutions of H1ModCH2 in 2D for comparisons with corresponding solutions of the Camassa-Holm equation in \cite{HS2013}. The comparisons of ModCH2 in 2D and 3D with corresponding solutions of the Camassa-Holm equation in \cite{HS2013} will be deferred to a later paper. In the later paper, we will also present comparisons of Camassa-Holm solutions in 2D and 3D with the corresponding solutions of EPDiff$(H_{div}$), as introduced in \cite{KSD2001}.

\end{remark}

\color{black}

\subsection{Simulations of emergent H1ModCH2 solutions}

In the rest of this section, we consider computational simulations of the H1ModCH2 equation \eqref{eq-H1ModCH2-ndim} dynamics in 2D. We will present five initial conditions in the paper and ten initial conditions in the supplementary materials. For each initial condition, we consider the dynamical exchange between kinetic and potential energy. This will be illustrated by starting with only kinetic energy with zero initial elevation in sections \ref{ssec: plate}, \ref{ssec: skew} and \ref{ssec: wedge}, or with only potential energy in \ref{ssec: single dam break} and \ref{ssec: dual dam break}. Each of five these initial conditions is simulated for two values of the interaction range, $\alpha$, relative to the characteristic width of the initial condition, $w_0$, which all of the simulations have in common. The two values of $\alpha$ selected in these simulations are given by $\alpha = w_0$ and $8\alpha = w_0$. For each of these five initial conditions and the two chosen values of $\alpha$, two figures with six panels are presented corresponding to the evolution of $|u|^2$ and $(\overline{D} - \overline{b}(\bs{x}))$, displayed as $6$ snapshots. The quantity $(\overline{D} - \overline{b}(\bs{x})) := G_{Q^{H1}_{op}(\alpha)}*({D}-{b}(\bs{x}))$ is the elevation, smoothed by convolution with the Green function $G_{Q^{H1}_{op}(\alpha)}$ defined in \eqref{eq-H1ModCH2-ndim}. For convenience, in this section we will refer to the smoothed difference $(\overline{D} - \overline{b}(\bs{x}))$ as the \emph{elevation}.

The first panel (top left panel) in each figure corresponds to the initial condition. The subsequent panels, reading across the first row and then across the second row, are snapshots at the subsequent times. The domain is $[0,2\pi]\times [0,2\pi]$ with doubly periodic boundary conditions. Coordinates are $x$ horizontally and $y$ vertically. We use the colour map shown in figure \ref{fig:velocity wave cmap} for the L2 norm of the velocity, $|u|^2$, where the minimal values and maximal values appear grey and white respectively. This is the same approach used in \cite{HS2013} where the black colour at $12.5\%$ intensity exists to show the outlines of spatially confined velocity segments. While the colour map \ref{fig:velocity wave cmap} is apt at showing small scale features for positive definite fields, it is not suitable to plot figures that take on negative values. Thus, the elevation figures will use the standard colour map \emph{turbo}. In each figure, the colour map is determined for each panel separately, so that the features of each snapshots are visible. The scales of the 2D plots are included in each figure alongside the colour bar such that the variation of the intensity across each panel is clear.
Four 1D slices of the domain are included in each snapshot in the directions shown in figure \ref{fig:1d profile}. Specifically, the solid black line is the profile along the horizontal $y = \pi$, the dashed red line is along the vertical $x = \pi$, the solid green line is along the upward diagonal $y = x - \pi$, and the dashed blue line is along the downward diagonal $y = \pi-x$. Similarly to the 2D snapshots, the scales of the 1D plots are also determined per panel for maximal clarity.
\begin{figure}[H]
    \centering
    \includegraphics[width=\textwidth, height=0.05\textwidth]{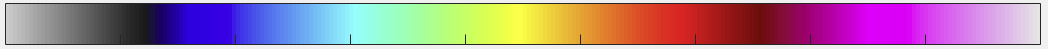}
    \caption{Custom colour map for the upcoming $|u|^2$ plots.}
    \label{fig:velocity wave cmap}
\end{figure}
\begin{figure}[H]
    \centering
    \includegraphics[width=0.3\textwidth, height=0.3\textwidth]{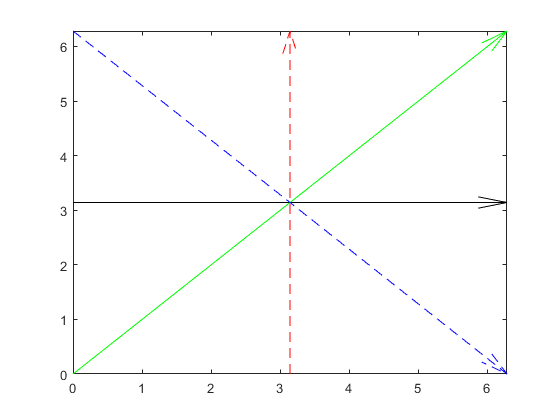}
    \caption{Locations of the 1D profiles of $|\bs{u}|$ and $\overline{D}$ in the 2D numerical simulations.}
    \label{fig:1d profile}
\end{figure}
The simulations apply a Fourier spectral method with $2048$ Fourier nodes in both the $x$ and $y$ directions. De-aliasing is accomplished by truncating the highest $\frac{1}{3}$ wave numbers. The time stepping method is the $5^{th}$ order Runge Kutta Fehlberg (RKF45) method with adaptive time stepping corresponding to the well known formula for step control,
\begin{align}
    h_i = \gamma h_{i-1}\left(\frac{\epsilon|h_{i-1}|}{|| \Bar{u}_i - \Hat{u}_i ||}\right)^{1/p}. \label{eq:stepsize}
\end{align}
In \eqref{eq:stepsize} we denote, as follows. At step $i$, we have the fourth order solution $\Bar{u}_i$ and fifth order solution $\Hat{u}_i$ as well as the previous time step $h_{i-1}$. The value $p = 4$ is the order of the solution $\Bar{u}_i$ and the order of $\hat{u}_i$ is $p+1$. If the L2 norm of $\Bar{u}_i - \Hat{u}_i$ is less than the tolerance $\epsilon$, the step size for the next time step is derived from \eqref{eq:stepsize}. If the $||\Bar{u}_i - \Hat{u}_i|| > \epsilon$, the current step is repeated with the step size derived from \eqref{eq:stepsize}. The relative tolerance and safety factor used in this work are $\epsilon = 10^{-5}$ and $\gamma = 0.9$ respectively.
\newpage 

\subsubsection{Plate}\label{ssec: plate}
In figures \ref{fig:snapshot plate u 1}-\ref{fig:snapshot plate rhobar 4}, we consider the combined dynamics of the velocity magnitude $|\bs{u}|$ and elevation $(\overline{D} - \overline{b}(\bs{x}))$ in the interplay of kinetic and potential energy for different values of $\alpha$, starting from the same initial conditions in a doubly periodic square domain with a flat bottom topography, so $\overline{b}(\bs{x})=const$. The Plate initial condition is inspired by the two SAR images in figure \ref{fig: Gibralter}. The first of these two SAR images shows the surface signature of an internal wave propagating midway through the Gibralter Strait. The second SAR image shows the train of wavefront surface signatures which develops after the internal wave has propagated into the open Mediterranean Sea. Initially, the momentum $\bs{m}$ shown in the first panel of the Plate figure is distributed along a line segment whose corresponding velocity falls off exponentially as $e^{-|x|/w_0}$ at either end of the segment and also in the transverse direction. Thus, the transverse slice of the fluid velocity profile shown in the rectangular strip below the panel as a black curve has a contact discontinuity, i.e., a jump in its derivative. The name ``Plate'' also refers to the corresponding case for the 2D CH dynamics simulated in \cite{HS2013}. The advected depth variable $D$ is initially at rest and the elevation is flat, so $\overline{D}(\bs{x},0)=const$. 

Figure \ref{fig:snapshot plate u 1} shows snapshots of the velocity profile of the initially rightward moving line segment. The support of the velocity solution develops a curvature and ``balloons'' outward as it moves rightward. It also stretches because the endpoints of its profile are fixed by the imposed exponential fall-off of velocity there. The shapes of the velocity profiles in the transverse direction of travel are shown by the 1D plots beneath the 2D snapshots. The bottom panels of \ref{fig:snapshot plate u 1} show the smoothing of the initial contact curves. Figure \ref{fig:snapshot plate rhobar 1} shows the snapshot of elevation $(\overline{D}(\bs{x},t) - \overline{b}(\bs{x}))$ accompanying the evolution of velocity in figure \ref{fig:snapshot plate u 1}. Note that the moving peak in elevation is accompanied by a trailing depression. This happens because of conservation of total mass. Namely, mass conservation implies that the moving surface elevation of an initially flat elevation profile must be accompanied by a corresponding moving depression of the elevation. The peak of the elevation follows the motion of the velocity profile. However, the profiles of velocity and elevation do not develop the same shape, because of the trailing depression below the mean elevation. The region of depression formed behind the peak extends from the initial position of the velocity profile to the tail of the current velocity profile. 

\paragraph{Wavefront emergence}
When $\alpha < w_0$, in figure \ref{fig:snapshot plate u 4}, the unstable initial velocity profile produces a train of peakon segments emerging as the initial profile breaks up. Each of the emergent wavefronts is curved because it velocity vanishes at the initial endpoints. The number of wavefronts depends on the size of $\alpha$. In figures \ref{fig:snapshot plate u 4}-\ref{fig:snapshot plate rhobar 4}, the first emitted velocity wavefronts have the highest velocity and subsequent wavefronts have lower velocity. Consequently, they will not overtake each other and a wave train will be formed. The material peaks travelling along with the velocity profiles also have the feature that the first peak is the highest and all subsequent peaks are lower. The depression region is now bounded by the location of initial velocity profile and the arc defined by the slowest emitted wavefront. The process of velocity wavefront emergence takes time to complete. This is shown in the last panel in figure \ref{fig:snapshot plate u 4}, where the initial condition has evolved into $6$ fully formed segments ahead of ramps. As time progresses further, the ramps will develop into a train of wavefront segments.

Figure \ref{fig:snapshot plate rhobar 4} shows the elevation associated to the velocity in figure \ref{fig:snapshot plate u 4}. Panel $1$ of figure \ref{fig:snapshot plate rhobar 4} shows the initially flat elevation. Panel 2 shows the early development of a wavetrain of positive elevation. As expected, the leading wave is the tallest. In panels $2$, $3$ and $4$ one sees the rightward propagation of mass as the wave train moves away from the initial rightward impulse. In the subsequent panels one sees the continued development of a leftward-moving depression of the surface due to the emission of the rightward moving wave train of positive elevation. The grey rectangular strips below the panels show details of the wave-forms along the colour-coded directions in figure \ref{fig:1d profile}. 
Note, however, that the elevation of the surface between the successive wavefronts in the wave train is \emph{less} that the initial level of the fluid at rest. Indeed, the depression is developing a counterflow in the opposite direction which might eventually cause a large scale oscillation in the wake of the plate. Comparing the properties of the fastest wavefronts for different values of $\alpha$, we see that both the material and velocity wavefronts are higher for smaller values of $\alpha$. This is due to conservation of mass and energy, in which $\alpha$ controls the width of the wave profile. 

\begin{comment}
\begin{framed}
{\color{blue}RH: Some description of stretching and time reversal,
\paragraph{Time reversal}
As the system is Hamiltonian, the motion is reversible. We test the reversibility of the numerical scheme in the different initial conditions by integrating the EP equations backwards in time for different values of $\alpha$. }
\end{framed}
\end{comment}

\foreach \y in {1,4}{
    	\pgfmathtruncatemacro{\alphaValue}{2^(\y -1 )}
    \begin{figure}[H]
    	\centering
        \foreach \x in {0,1,...,5}{
    	    \begin{subfigure}[b]{0.3\textwidth}
    			\centering
    			\includegraphics[width=\textwidth, height=\textwidth]{figures/one_plate/\y/\x_usqr_comp.jpg}
    		\end{subfigure}
    	}   
	\ifnum \alphaValue = 1
       	\caption{\small Evolution of $|u|^2$ with the ``plate'' initial condition with $\alpha = w_0$} 
       	\else
       	\caption{\small Evolution of $|u|^2$ with the ``plate'' initial condition with $\alpha = w_0/\alphaValue$} 
       	\fi
    	\label{fig:snapshot plate u \y}
    \end{figure}
    
    \begin{figure}[H]
    	\centering
        \foreach \x in {0,1,...,5}{
    	    \begin{subfigure}[b]{0.3\textwidth}
    			\centering
    			\includegraphics[width=\textwidth, height=\textwidth]{figures/one_plate/\y/\x_elev_comp.jpg}
    		\end{subfigure}
    	} 
    	\ifnum \alphaValue = 1
       	\caption{\small Evolution of $\overline{D}-\overline{b}$ with the ``plate'' initial condition with $\alpha = w_0$. } 
       	\else
       	\caption{\small Evolution of $\overline{D}-\overline{b}$ with the ``plate'' initial condition with $\alpha = w_0/\alphaValue$. } 
       	\fi
    	\label{fig:snapshot plate rhobar \y}
    \end{figure}
    % \clearpage
}
\newpage 

%%% Skew simulations
\subsubsection{Skew}\label{ssec: skew}
Skew flows in figures \ref{fig:snapshot skew u 4} and \ref{fig:snapshot skew rhobar 4} are initiated with two peakon segments of the same width and with constant elevation. The peakon segment located at the back has $1.5$ times the amplitude of the peakon segment moving horizontally. Thus, the waves emerging from the back peakon will overtake the waves emerging from the peakon moving to the right by moving along the negative diagonal. Panel $2$ of figure \ref{fig:snapshot skew u 4} shows the result of collisions of the first emitted curved velocity segments. Here, both overtaking and head-on collisions have occurred along different axes and the resulting non-linear transfer of momentum has resulted in the merging, or \emph{reconnection}, of the wave segments. The collision has also produced a \emph{hotspot} of momentum and elevation located at the intersection point. This hotspot expands rapidly outward to form the red region of the rightmost wavefront in panel $3$. The appearance of hotspots during the reconnection of wavefronts is also seen in the dynamics of doubly periodic solutions of the  Kadomtsev–Petviashvili equation \cite{CK2009} and also observed, for example, in a famous photograph of crossing swells in the Atlantic Ocean \cite{S2012}. 

The notion of Lagrangian \emph{memory wisps} introduced in \cite{HS2013} is particularly visible in panel $3$ of the elevation evolution in figure \ref{fig:snapshot skew rhobar 4}, where two wisps can be found connecting the boundaries of the expanding hotspot to edge points of elevation segments. By examining intermediate snapshots, we see that the initial memory wisp connects from the hotspot to the edge of the elevation segment after the collision travelling downwards. Via hotspot expansion and the emission of additional wavefronts from the inital conditions, the wisp splits into two and connects to different elevation segments. In panel $4$ to $6$, we see the same interaction of subsequently emitted wavefronts with multiple collisions and reconnections. In each of the collisions, memory wisps are produced as seen between resultant wavefronts which suggests the hotspots are part of the mechanism creating the wisps. We note that the persistent memory wisps in panel $6$ between the most and second most rightward elevation segments are the same memory wisps seen in panel $3$. This suggests that the memory wisps are not produced by the numerical method, instead they are products of the wavefront collisions which preserve the reversibility of the evolution.
\newpage 

\foreach \y in {1,4}{
    	\pgfmathtruncatemacro{\alphaValue}{2^(\y -1 )}
\begin{figure}[H]
	\centering
    \foreach \x in {0,1,...,5}{
	    \begin{subfigure}[b]{0.3\textwidth}
			\centering
			\includegraphics[width=\textwidth, height=\textwidth]{figures/skew/\y/\x_usqr_comp.jpg}
		\end{subfigure}
	}  
	\ifnum \alphaValue = 1
       	\caption{\small Evolution of $|u|^2$ with the ``skew'' initial condition with $\alpha = w_0$.} 
       	\else
   	\caption{\small Evolution of $|u|^2$ with the ``skew'' initial condition with $\alpha = w_0/8$.} 
   	\fi
	\label{fig:snapshot skew u \y}
\end{figure}

\begin{figure}[H]
	\centering
    \foreach \x in {0,1,...,5}{
	    \begin{subfigure}[b]{0.3\textwidth}
			\centering
			\includegraphics[width=\textwidth, height=\textwidth]{figures/skew/\y/\x_elev_comp.jpg}
		\end{subfigure}
	} 	
	\ifnum \alphaValue = 1
       	\caption{\small Evolution of $\overline{D}-\overline{b}$ with the ``skew'' initial condition with $\alpha = w_0$.} 
       	\else
   	\caption{\small Evolution of $\overline{D}-\overline{b}$ with the ``skew'' initial condition with $\alpha = w_0/8$.} 
   	\fi
	\label{fig:snapshot skew rhobar \y}
\end{figure}
% \clearpage
}
%%% Wedge

\newpage

\subsubsection{Wedge} \label{ssec: wedge} 
The ``wedge'' initial condition is a modification of the skew collision in which the initial upper peakon segment travels downward along the negative diagonal and the initial lower peakon segment travels upward along the positive diagonal. The magnitudes of the velocities are the same and there is a reflection symmetry along the horizontal axis in the middle of the domain, along $y = \pi$. When the emergent wavefronts meet along $y=\pi$, their vertical momentum components collide in opposite directions (head-on). The ``wedge'' initial condition can be seen on the left of the lower panels of figure \ref{fig:snapshot wedge u 1}, emerging from the line $y=\pi$. In panel $3$ of figure \ref{fig:snapshot wedge u 1}, the collision of the velocity segments forms a hotspot along the mid-line, which expands outward away during the reconnection process near the center of panel $4$. As these hotspots expand further in the next panels, they leave behind memory wisps in the velocity which are visible in panel $6$. These memory wisps are not seen, though, in the snapshots of elevation in figure \ref{fig:snapshot wedge rhobar 1}, as they are obscured near the boundaries between the depression regions and the elevation of the material wave segments.

For $w_0 = 8\alpha$ in figures \ref{fig:snapshot wedge u 4} and \ref{fig:snapshot wedge rhobar 4}, multiple ``wedge'' collisions occur from the wave train emerging from the initial conditions. They also interact among wavefronts from the same wave train due to the elastic collision property. This interaction produces fast, small-scale oscillations which resemble the emergent wavefronts, broken into even shorter ``shards'' seen in panel $3$ to $6$ of figure \ref{fig:snapshot wedge u 4} and in \ref{fig:snapshot wedge rhobar 4}. These broken shards of wave segments arise when the numerical method can no longer resolve the smallest scale behaviour. Lowering the values of $\alpha$ narrows both the velocity and elevation wave segments. It also has the effect of highlighting the presence of memory wisps, as more collisions occur with higher transfer of momentum and thus greater separation of the wavefronts, as seen in the panel $6$ of figure \ref{fig:snapshot wedge rhobar 4}. 

Other aspects of head-on collisions will be discussed next in section \ref{ssec: parallel} for the ``parallel'' initial conditions. 

\newpage 

\foreach \y in {1,4}{
    	\pgfmathtruncatemacro{\alphaValue}{2^(\y -1 )}
    \begin{figure}[H]
    	\centering
        \foreach \x in {0,1,...,5}{
    	    \begin{subfigure}[b]{0.3\textwidth}
    			\centering
    			\includegraphics[width=\textwidth, height=\textwidth]{figures/wedge/\y/\x_usqr_comp.jpg}
    		\end{subfigure}
    	}   
	\ifnum \alphaValue = 1
       	\caption{\small Evolution of $|u|^2$ with the ``wedge'' initial condition with $\alpha = w_0$.} 
       	\else
       	\caption{\small Evolution of $|u|^2$ with the ``wedge'' initial condition with $\alphaValue\alpha = w_0$.} 
       	\fi
    	\label{fig:snapshot wedge u \y}
    \end{figure}
    
    \begin{figure}[H]
    	\centering
        \foreach \x in {0,1,...,5}{
    	    \begin{subfigure}[b]{0.3\textwidth}
    			\centering
    			\includegraphics[width=\textwidth, height=\textwidth]{figures/wedge/\y/\x_elev_comp.jpg}
    		\end{subfigure}
    	} 
    	\ifnum \alphaValue = 1
       	\caption{\small Evolution of $\overline{D}-\overline{b}$ with the ``wedge'' initial condition with $\alpha = w_0$.} 
       	\else
       	\caption{\small Evolution of $\overline{D}-\overline{b}$ with the ``wedge'' initial condition with $\alphaValue\alpha = w_0/$.} 
       	\fi
    	\label{fig:snapshot wedge rhobar \y}
    \end{figure}
    % \clearpage
}

%%% Parallel 

\newpage 

\subsubsection{Parallel} \label{ssec: parallel} 
The initial condition for the ``parallel'' collision comprises two peakon segments of equal and opposite magnitudes moving toward each other along vertically offset parallel horizontal lines, as shown in figure \ref{fig:snapshot parallel u 1}. This situation differs from the overtaking (rear-end) collisions seen in the ``skew'' initial conditions, as the collisions are head-on; so they involve wavefronts with positive and negative velocity components. In 1D, when the wavefronts are peakons, no vertical offset can occur and an antisymmetric initial condition on the real line produces a collision in which the two weak solutions bounce off each other elastically in opposite directions. In the 2D case, the offset initial condition introduces angular momentum into the system. Consequently, the offset head-on collision can access angular degrees of freedom and thus it will show more complex behaviour than the head-on collision in 1D.

Consider the case where $\alpha = w_0$ in figure \ref{fig:snapshot parallel u 1}. The initial velocity segments balloon outwards and the shape is smoothed as occurs in the ``plate'' condition. When the wavefronts collide in panel $3$, the magnitude of velocity along the collision front vanishes and the velocity profile becomes very steep as seen also in 1D peakon collisions. In panel $4$, we see that the wavefront segments which did not undergo head-on collisions contain hotspots. The hotspots indicate where reconnections have occurred. These hotspots expand in panels $5$ and $6$ into a velocity profile which balloons outwards with an angle away from the vertical axes. The results of the head-on collisions are the dark segments connecting the upper and lower velocity wavefronts. The scattering angle seen clearly in the third panel of figure \ref{fig:snapshot parallel rhobar 1} is due to the conservation of angular momentum during the offset head-on collision.

Figure \ref{fig:snapshot parallel rhobar 1} shows snapshots of the elevation during the offset head-on collision. As the elevation segments are advected with the velocity profile, we see an elevation head-on collision in panel $3$. In contrast with the velocity profile where the velocity is tends to zero along the collision front, the elevations are rising in the collision to create a elevation segment of large amplitude. This reinforced elevation then decreases in height in panel $4$ to $6$ as the elevation wavefront emerges from the head-on collisions. This is clearest from the black 1D profile in the grey rectangular strip below panel $6$.

When $\alpha < w_0$, the evolution becomes even more complex because entire trains of wavefronts are involved, as seen in figures \ref{fig:snapshot parallel u 4} and \ref{fig:snapshot parallel rhobar 4}. In these figures, one see the reconnections of velocity and elevation segments which had undergone head-on collisions with those segments that had not collided. The complexity builds, as the head-on collisions and reconnections recur again and again, as additional wavefronts continue to be emitted from the initial conditions.

\newpage

\foreach \y in {1,4}{
    	\pgfmathtruncatemacro{\alphaValue}{2^(\y -1 )}
    \begin{figure}[H]
    	\centering
        \foreach \x in {0,1,...,5}{
    	    \begin{subfigure}[b]{0.3\textwidth}
    			\centering
    			\includegraphics[width=\textwidth, height=\textwidth]{figures/parallel/\y/\x_usqr_comp.jpg}
    		\end{subfigure}
    	}   
	\ifnum \alphaValue = 1
       	\caption{\small Evolution of $|u|^2$ with the ``parallel'' initial condition with $\alpha = w_0$.}
       	\else
       	\caption{\small Evolution of $|u|^2$ with the ``parallel'' initial condition with $\alphaValue\alpha = w_0$.} 
       	\fi
    	\label{fig:snapshot parallel u \y}
    \end{figure}
    
    \begin{figure}[H]
    	\centering
        \foreach \x in {0,1,...,5}{
    	    \begin{subfigure}[b]{0.3\textwidth}
    			\centering
    			\includegraphics[width=\textwidth, height=\textwidth]{figures/parallel/\y/\x_elev_comp.jpg}
    		\end{subfigure}
    	} 
    	\ifnum \alphaValue = 1
       	\caption{\small Evolution of $\overline{D}-\overline{b}$ with the ``parallel'' initial condition with $\alpha = w_0$.} 
       	\else
       	\caption{\small Evolution of $\overline{D}-\overline{b}$ with the ``parallel'' initial condition with $\alphaValue\alpha = w_0$.} 
       	\fi
    	\label{fig:snapshot parallel rhobar \y}
    \end{figure}
    % \clearpage
}

%%% Double dam breaks

\newpage 

\subsubsection{Dam Break}\label{ssec: single dam break}
A Dam Break, or Lock Release, flow is produced when at time $t=0$ a volume of fluid at rest behind a dam, or lock, is suddenly released. Gravity then drives the flow, as potential energy is converted in kinetic energy. Here, we treat the case of a radially symmetric Gaussian distribution of initial depth 
\[\overline{D} = 2e^{-((x-\pi)^2 + (y-\pi)^2)/{w_0}} + b(\bs{x}),\] with constant, non zero bathymetry, $b(\bs{x}) = const > 0$. This corresponds to the case where the radially symmetric Gaussian distribution of initial surface elevation is released into a fluid at rest with a flat surface over a constant bathymetry. Consider the elevation profile $\overline{D}-\overline{b}$ in figure \ref{fig:snapshot dam break rhobar 1}. The first panel shows the initial condition for the elevation, while the second panel show the plateauing of the elevation peak in lowering of its initial Gaussian profile which, in turn, becomes wider as mass is pushed outward radially by gravity. When the critical width $\alpha$ of the expansion is reached, then a wavefront is emitted radially outwards on the left and right hand sides of the domain, as seen in panel $3$ and $4$. We note that the elevation becomes negative behind the formation of the material wavefront and it becomes more negative in subsequent panels after panel $3$. The leading edge of the wavefront subsides exponentially as it evolves, as does the shape of the leading edge of the velocity wavefront in figure \ref{fig:snapshot dam break u 1}. In the velocity profile, we see that the emerging wavefront takes the peakon form in the third panel. As the system evolves, the leading edge of the velocity wavefront is similar in shape to the material wave profiles in panels three, four and five.

For smaller values of $\alpha$, a train of velocity and material wavefronts rapidly develops, as seen in figures \ref{fig:snapshot dam break u 4} - \ref{fig:snapshot dam break rhobar 4} where $w_0 = 8\alpha$. Similarly to the ``plate'' initial condition, the first wavefront has the highest velocity and elevation, while subsequent wavefronts have lower velocity and elevation. The elevation ahead of the front of the first wavefront remains flat, but as the expansion continues the level of the fluid surface drops behind the expanding wave train. If one looks closely, one sees that the level of the surface between the wavefronts in the wave train is lower that the initial level at rest. Perhaps this would eventually produce a counter flow. 

Now we consider a variation of the Dam break initial condition that does produce persistent velocity peakon wavefronts. The initial conditions for the figures \ref{fig:snapshot dam break no bathymatry rhobar 1} - \ref{fig:snapshot dam break no bathymatry u 4} are \[\overline{D} = 2e^{-((x-\pi)^2 + (y-\pi)^2)/{w_0}} + b(\bs{x}),\] with $b(\bs{x}) = 0$ and $\bs{u} = 0$. This corresponds to the case where the initial surface elevation is released into the constant bathymetry. Comparing the velocity profile for the same value of $\alpha$ in figure \ref{fig:snapshot dam break no bathymatry u 1} and \ref{fig:snapshot dam break u 1}, we see that the evolution are very similar for the first three panels as the development of peakon wavefronts. However, we see in the bottom three panels of figure \ref{fig:snapshot dam break no bathymatry u 1} that the peakons persist and travel outwards radially in a wave train with decreasing amplitude. The elevation profile in figure \ref{fig:snapshot dam break no bathymatry rhobar 1} also start with the plateauing and lowering of the initial Gaussian elevation in panel two. In panel three, one sees the start of the formation of a material wavefront. Instead of the wavefront being formed and emitted like the velocity wavefront, the material wavefront travels outwards, loses amplitude and vanishes when it reaches the edge of the elevation distribution. The width of the elevation distribution is widened in this process, as seen in the bottom three panels. 

Similarly, for smaller values of $\alpha$, a train of velocity and material wavefronts appears from the initial condition, as seen in figures \ref{fig:snapshot dam break no bathymatry rhobar 4} - \ref{fig:snapshot dam break no bathymatry u 4} where $w_0 = 8\alpha$. The process of formation and annihilation of material peakons persists and the shape of the wavefront reassemble the peakon shape for more subsequent waves in the wave train. Since the elevation is not negative in this flow, no counter flow would be produced. 

From these two variations of the Dam Break problem we see that from a smooth, spatial confined initial conditions, the time varies for persistent, peakon-shaped wavefronts to develop. However, the issue of the time of formation of peakon wave profiles is beyond the scope addressed in the present paper. 
\foreach \y in {1,4}{
    	\pgfmathtruncatemacro{\alphaValue}{2^(\y -1 )}
    \begin{figure}[H]
    	\centering
        \foreach \x in {0,1,...,5}{
    	    \begin{subfigure}[b]{0.3\textwidth}
    			\centering
    			\includegraphics[width=\textwidth, height=\textwidth]{figures/dam break/\y/\x_elev_comp.jpg}
    		\end{subfigure}
    	} 
    	\ifnum \alphaValue = 1
       	\caption{\small Evolution of $\overline{D}-\overline{b}$ with the ``dam break'' initial condition with $\alpha = w_0$.} 
       	\else
       	\caption{\small Evolution of $\overline{D}-\overline{b}$ with the ``dam break'' initial condition with $\alphaValue\alpha = w_0$.} 
       	\fi
    	\label{fig:snapshot dam break rhobar \y}
    \end{figure}
    
    \begin{figure}[H]
    	\centering
        \foreach \x in {0,1,...,5}{
    	    \begin{subfigure}[b]{0.3\textwidth}
    			\centering
    			\includegraphics[width=\textwidth, height=\textwidth]{figures/dam break/\y/\x_usqr_comp.jpg}
    		\end{subfigure}
    	}   
	\ifnum \alphaValue = 1
       	\caption{\small Evolution of $|u|^2$ with the ``dam break'' initial condition with $\alpha = w_0$.} 
       	\else
       	\caption{\small Evolution of $|u|^2$ with the ``dam break'' initial condition with $\alphaValue\alpha = w_0$.} 
       	\fi
    	\label{fig:snapshot dam break u \y}
    \end{figure}
}

%%% Plots for dam break with zero bathymetry
\begin{figure}[H]
	\centering
    \foreach \x in {0,1,...,5}{
	    \begin{subfigure}[b]{0.3\textwidth}
			\centering
			\includegraphics[width=\textwidth, height=\textwidth]{figures/dam break/1_0/\x_elev_comp.jpg}
		\end{subfigure}
	} 
   	\caption{\small Evolution of $\overline{D}-\overline{b}$ with the ``dam break'' initial condition with $\alpha = w_0$ and zero bathymetry.} 
	\label{fig:snapshot dam break no bathymatry rhobar 1}
\end{figure}
\begin{figure}[H]
	\centering
    \foreach \x in {0,1,...,5}{
	    \begin{subfigure}[b]{0.3\textwidth}
			\centering
			\includegraphics[width=\textwidth, height=\textwidth]{figures/dam break/1_0/\x_usqr_comp.jpg}
		\end{subfigure}
	}   
   	\caption{\small Evolution of $|u|^2$ with the ``dam break'' initial condition with $\alpha = w_0$  and zero bathymetry.} 
	\label{fig:snapshot dam break no bathymatry u 1}
\end{figure}

\begin{figure}[H]
	\centering
    \foreach \x in {0,1,...,5}{
	    \begin{subfigure}[b]{0.3\textwidth}
			\centering
			\includegraphics[width=\textwidth, height=\textwidth]{figures/dam break/4_0/\x_elev_comp.jpg}
		\end{subfigure}
	} 
   	\caption{\small Evolution of $\overline{D}-\overline{b}$ with the ``dam break'' initial condition with $8\alpha = w_0$.} 
	\label{fig:snapshot dam break no bathymatry rhobar 4}
\end{figure}

\begin{figure}[H]
	\centering
    \foreach \x in {0,1,...,5}{
	    \begin{subfigure}[b]{0.3\textwidth}
			\centering
			\includegraphics[width=\textwidth, height=\textwidth]{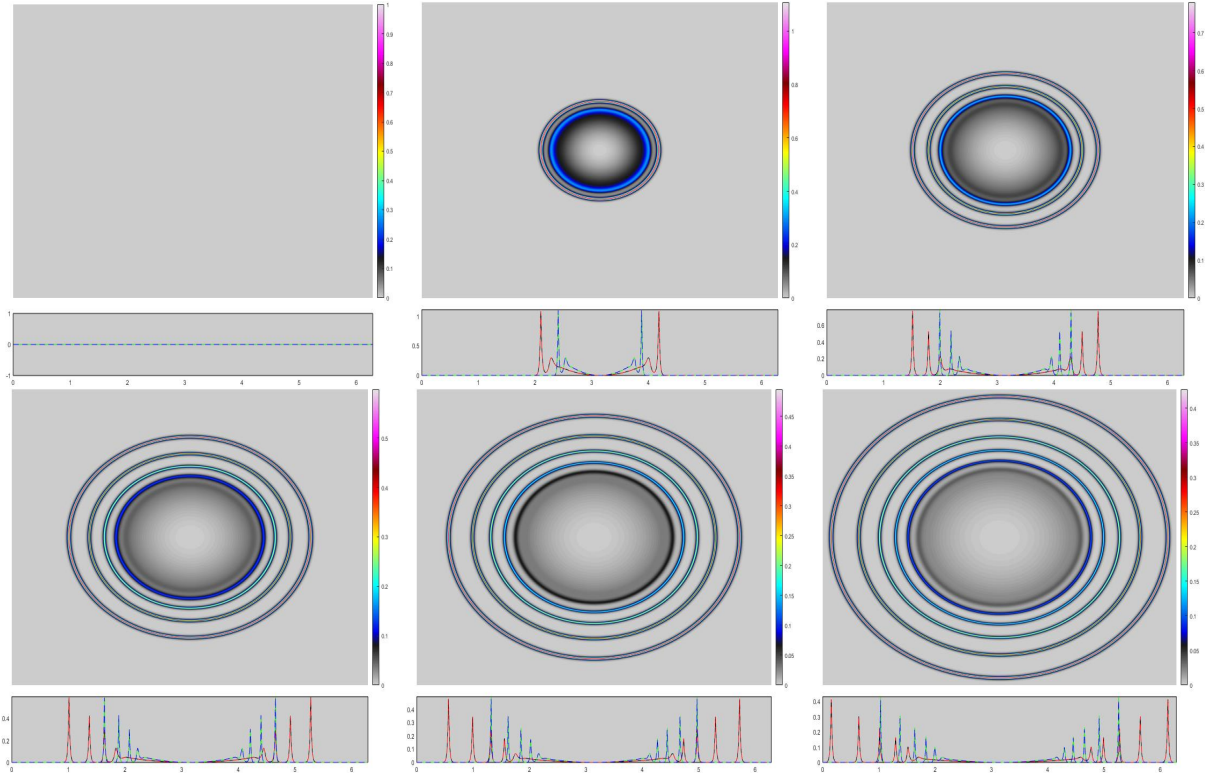}
		\end{subfigure}
	}   
   	\caption{\small Evolution of $|u|^2$ with the ``dam break'' initial condition with $8\alpha = w_0$.} 
	\label{fig:snapshot dam break no bathymatry u 4}
\end{figure}
\clearpage

%%%%%% Dual dam break
\subsubsection{Dual Dam Break} \label{ssec: dual dam break}
Here, we treat the Dam Break flow in which the initial condition contains two radially symmetric Gaussian distributions of initial surface elevation. These are simultaneously released into a fluid at rest with a flat surface over a constant bathymetry. To study the interaction of both velocity and elevation wavefronts, we consider the case where the bathymetry $b > 0$. The emergence property of wavefronts remains the same of the single Dam Break in section \ref{ssec: single dam break}. In figure \ref{fig:snapshot dual dam break rhobar 1}, panel $1$ is the initial condition and panel $2$ is the emergence of elevation wavefronts. In the middle of the panel $3$, one sees the head-on collisions of these emergent wavefronts. In the center of the domain in panel $4$ to $6$, one sees the head-on collisions of the emitted radial peakons and their reconnections in the form of two rapidly expanding hotspots located along $x = \pi$. As one part of the elevation wavefront expands radially away from the center of the domain, it leaves a widening region of depression behind it which creates a counter-flow, which one sees developing in panel $6$ as the dark purple region.
The corresponding velocity profile in figure \ref{fig:snapshot dual dam break u 1} evolves similarly to the ``Parallel'' and ``Single Dam Break'' flows for head-on collisions and emergence of wavefronts respectively.   

For smaller values of $\alpha$, a train of peakon wavefronts rapidly ensues, as seen in figures \ref{fig:snapshot dual dam break u 4} - \ref{fig:snapshot dual dam break rhobar 4} where $w_0 = 8\alpha$. Consider the interaction in the centre of the domain after the initial head-on collision of the first wave in the emergent wave train in panel $3$. Panel $4$ shows the ``rebound'' wavefront interacting with the subsequent wavefront in the wave train to create hotspots above and below $x = \pi$. This process repeats for every wavefront and creates a checkboard pattern in the region above and below $x = \pi$. The connection between wavefronts are the memory wisps, seen in panel $5$ and $6$. Thus, this interaction produces a cellular elevation profile which is locally similar to that of the doubly periodic cnoidal waves seen in solutions of the Kadomtsev–Petviashvili equation \cite{S2012,CK2009}. 

\newpage 

\foreach \y in {1,4}{
    	\pgfmathtruncatemacro{\alphaValue}{2^(\y -1 )}
    \begin{figure}[H]
    	\centering
        \foreach \x in {0,1,...,5}{
    	    \begin{subfigure}[b]{0.3\textwidth}
    			\centering
    			\includegraphics[width=\textwidth, height=\textwidth]{figures/dual dam break/\y/\x_usqr_comp.jpg}
    		\end{subfigure}
    	}   
	\ifnum \alphaValue = 1
       	\caption{\small Evolution of $|u|^2$ with the ``dual dam break'' initial condition with $\alpha = w_0$.} 
       	\else
       	\caption{\small Evolution of $|u|^2$ with the ``dual dam break'' initial condition with $\alphaValue\alpha = w_0$.} 
       	\fi
    	\label{fig:snapshot dual dam break u \y}
    \end{figure}
    
    \begin{figure}[H]
    	\centering
        \foreach \x in {0,1,...,5}{
    	    \begin{subfigure}[b]{0.3\textwidth}
    			\centering
    			\includegraphics[width=\textwidth, height=\textwidth]{figures/dual dam break/\y/\x_elev_comp.jpg}
    		\end{subfigure}
    	} 
    	\ifnum \alphaValue = 1
       	\caption{\small Evolution of $\overline{D}-\overline{b}$ with the ``dual dam break'' initial condition with $\alpha = w_0$.} 
       	\else
       	\caption{\small Evolution of $\overline{D}-\overline{b}$ with the ``dual dam break'' initial condition with $\alphaValue\alpha = w_0$.} 
       	\fi
    	\label{fig:snapshot dual dam break rhobar \y}
    \end{figure}
    \clearpage
}

%%%%%%%%%%%%%%%%%%%%%%%%%%%%%%%%%%%%%%%%%%%%%%%%%%%%%
\begin{comment}[Summary]$\,$

\noindent
Above, we have discussed the following additional points. 

\begin{itemize}
    \item Minimal description of a fluid.
    \item Emergent singular solutions for fluid momentum and advected quantities.
    \item Interaction of curves, their collisions and reconnections.
    \item Emergent singular solutions from smooth initial conditions.
\end{itemize}
\end{comment}
%%%%%%%%%%%%%%%%%%%%%%%%%%%%%%%%%%%%%%%%%%%%%%%%%%%%%

\section{Conclusions and outlook}
    Inspired by the SAR images of sea-surface wavefronts regarded as the signatures of the dynamics of internal waves propagating below the surface, we proposed in the introduction to derive a single-layer minimal model of the surface velocity and elevation whose solution behaviour would mimic the dynamics of the curved wavefronts seen in the SAR images, figure \ref{fig: Gibralter} and \ref{fig: South China Sea}. The computationally simulated solutions of the H1ModCH2 minimal model illustrated the emergence of trains of wavefronts which evolved into complex patterns as they propagated away from localised disturbances of equilibrium and interacted nonlinearly with each other through collisions, stretching and reconnection. 
    
%%%%%%%%%%%%%%%%%%%%%%%%%%%%%%%%%%%%%%%%%%%%%%%%%%%%%
    To mimic the wave-current interaction that drives the curved wavefronts seen in SAR images, we investigated a variety of computational simulation scenarios which addressed two questions. First, we asked how would an initial condition evolve if there were a current possessing kinetic energy, but the surface were flat and, thus, had no gravitational potential energy? This question was answered in sections \ref{ssec: plate}, \ref{ssec: skew} and \ref{ssec: wedge} to which we refer for details. Second, we considered the converse question for initial conditions in which a stationary elevation was released into still water, as discussed in sections \ref{ssec: single dam break} and \ref{ssec: dual dam break}. In addition, we have also mimicked the reconnection properties of the internal waves signatures in the cases where wavefront collisions occur.
    
    %Many of these initial configurations were inspired by the earlier work of \cite{HS2013} for the 2D EPDiff equation with only kinetic energy given by the $H^1$ norm of the velocity. Finally, t
    %It seemed plausible that singular solutions of 2D ModCH2 would emerge from these smooth initial conditions was made , because we had already known about the singular momentum map discovered for the $N$ dimensional ModCH2 equation in \cite{HT2009} and about the 1D simulations which had investigated its emergent singular solutions in \cite{HOT2009}. 
    
    We had hoped that the singular momentum map solutions discussed in section \ref{ssec: singmomap} would emerge from our simulations of wavefront trains arising from localised disturbances. This would have reduced the problem of  wave-current interaction among sea-surface wavefronts to the much simpler problem of mutual advection among curves in the plane. This would have been the case, of course, if we had started with the singular momentum map in equation \eqref{sing-soln-thm} as the solution ansatz which would follow the dynamics in \eqref{singsoldynamics}. However, we hoped to see wavetrains of singular solutions on embedded curve segments emerge from generic smooth confined initial conditions. In fact, we did see that effect in some of the simulations. We saw that wavetrains of peakon curves did form in some cases of our suite of simulated energy exchange dynamics, more specifically, the dam break problem with zero bathymetry in section \ref{ssec: single dam break}. However, in some other cases such as the ``plate'' in section \ref{ssec: plate}, the wavetrains of peakon curves did not form completely. That is, the singular solutions supported on embedded curves did not always form completely during the time intervals of our simulations. Moreover, in the ``dam break'' initial condition in section \ref{ssec: single dam break}, the leading peaks in the elevation began to form, and then slowly ebbed away and disappeared as other peaks emerged behind them and then disappeared later, as well. 
    
    So, the question of emergence of the singular momentum map solutions in \eqref{sing-soln-thm} for ModCH2 from a smooth confined initial condition dynamics remains open. In particular, the question is, under what conditions will a solution of the 2D ModCH2 equations starting from a smooth confined initial condition indeed produce a train of singular peakon curve segments, if ever?

\section*{Acknowledgements}
This paper was written in honour of Tony Bloch's 65th birthday. Happy birthday, Tony! Best wishes for many more fruitful happy years of collaboration with your friends in the geometric mechanics community. 
We would also like to thank our other friends and colleagues who have generously offered their attention, thoughts and encouragement in the course of this work during the time of COVID-19. 
%We are particularly grateful to E. Luesink, S. Takao, W. Pan, D. Crisan, O. Street, F. Gay-Balmaz, E. M\'emin, B. Chapron,  C. Franzke, J. C. McWilliams, B. Fox-Kemper, A. J. Roberts and W. Bauer for thoughtful comments and discussions.
DH is grateful for partial support from ERC Synergy Grant 856408 - STUOD (Stochastic Transport in Upper Ocean Dynamics). RH is supported by an EPSRC scholarship [grant number EP/R513052/1].

%\section*{Data availability} 
%No data was created or used in writing this paper.

%\input{Appendix.tex}

\end{document}